\documentclass[manuscript]{acmart}

\acmJournal{TOIS}
\acmDOI{}
\setcopyright{none}
\renewcommand\footnotetextcopyrightpermission[1]{}

\AtBeginDocument{%
  }

\usepackage{amsmath,amsfonts}
\usepackage{algorithmic}
\usepackage{algorithm}
\usepackage{array}
\usepackage{textcomp}
\usepackage{verbatim}
\usepackage{graphicx}
\usepackage{wrapfig}
\usepackage{amsthm}
\usepackage{dsfont}
\usepackage{xcolor}
\usepackage{booktabs}
\usepackage{makecell}
\usepackage{tabularx}
\usepackage{multirow}
\usepackage{multicol}
\usepackage{scalerel}
\usepackage{tikz}
\usepackage[inline]{enumitem}
\usepackage{pgfplots}
\usepackage{pgfplotstable}
\usepackage{pifont}
\usepackage{cleveref}

\providecommand{\mathbbm}[1]{\mathds{#1}}

\usetikzlibrary{plotmarks}
\usetikzlibrary{patterns}
\usetikzlibrary{pgfplots.statistics}
\pgfplotsset{compat=1.18}

\newtheorem{theorem}{Theorem}
\newtheorem{problem}{Problem}
\theoremstyle{remark}
\newtheorem{definition}{Definition}[section]

\newtheorem{lemma}{Lemma}[section]

\newtheoremstyle{problemstyle}
        {3pt}
        {3pt}
        {\normalfont}
        {1em}
        {\bfseries\itshape}
        {\normalfont\bfseries:}
        {.5em}
        {}
\theoremstyle{problemstyle}

\newtheorem{innercustomgeneric}{\customgenericname}
\providecommand{\customgenericname}{}
\newcommand{\newcustomtheorem}[2]{%
  \newenvironment{#1}[1]
  {%
   \renewcommand\customgenericname{#2}%
   \renewcommand\theinnercustomgeneric{##1}%
   \innercustomgeneric
  }
  {\endinnercustomgeneric}
}
\newcustomtheorem{customthm}{Theorem}
\newcustomtheorem{customlemma}{Lemma}

\newcommand{\circledchar}[2][gray!70]{%
    \tikz[baseline=(char.base)]{
        \node[shape=circle,draw=#1,fill=#1,text=white,inner sep=0.75pt,scale=0.85] (char) {#2};
    }%
}
\newcommand{\Method}{\texttt{PreGress}}
\newcommand{\cmark}{\textcolor{red}{\ding{51}}}
\newcommand{\xmark}{\ding{55}}

\begin{document}

\title{PreGress: Ranking-Native Pre-training and Prompting for Graph Node Ranking}

\author{Lujie Ban}
\email{lujieban@link.cuhk.edu.cn}
\affiliation{%
  \institution{School of Data Science, The Chinese University of Hong Kong, Shenzhen}
  \city{Shenzhen}
  \country{China}}

\author{Jiasheng Shi}
\email{shijiasheng@cuhk.edu.cn}
\affiliation{%
  \institution{School of Data Science, The Chinese University of Hong Kong, Shenzhen}
  \city{Shenzhen}
  \country{China}}

\author{Yingli Zhou}
\email{yinglizhou@link.cuhk.edu.cn}
\affiliation{%
  \institution{School of Data Science, The Chinese University of Hong Kong, Shenzhen}
  \city{Shenzhen}
  \country{China}}

\author{Kaiwen Xue}
\email{xuekaiwen6@huawei.com}
\affiliation{%
  \institution{GTS, Huawei Technologies Co., Ltd.}
  \country{China}}

\author{Daiyin Wang}
\email{wangdaiyin@huawei.com}
\affiliation{%
  \institution{GTS, Huawei Technologies Co., Ltd.}
  \country{China}}

\author{Xubin Li}
\email{lixubin@huawei.com}
\affiliation{%
  \institution{GTS, Huawei Technologies Co., Ltd.}
  \country{China}}

\author{Shuanghua Li}
\email{leesantwa@huawei.com}
\affiliation{%
  \institution{GTS, Huawei Technologies Co., Ltd.}
  \country{China}}

\author{Chenhao Ma}
\authornote{Corresponding author.}
\email{machenhao@cuhk.edu.cn}
\affiliation{%
  \institution{School of Data Science, The Chinese University of Hong Kong, Shenzhen}
  \city{Shenzhen}
  \country{China}}

\renewcommand{\shortauthors}{Ban et al.}

\begin{abstract}
Node ranking is a fundamental problem in graph information retrieval, measuring the relative importance of nodes and supporting a wide range of applications such as influence analysis, recommendation, and graph-based retrieval augmented generation.
However, exact computation of graph-based ranking measures is often computationally prohibitive at scale.
Existing GNN-based ranking methods provide scalable approximations, but they are typically tailored to individual ranking criteria and require retraining for each downstream task, which limits their transferability and efficiency.
Recent graph pre-training approaches aim to enable knowledge transfer across tasks, yet their learning objectives are largely misaligned with node ranking, resulting in suboptimal adaptability to ranking-oriented applications.
To address these limitations, we propose PreGress, the first ranking-native pre-training and prompting framework for supporting a wide range of node ranking tasks.
PreGress performs multi-task pre-training using our carefully designed objectives, including degree centrality prediction and attribute reconstruction, to jointly capture structural and attribute information.
To support heterogeneous ranking criteria, we design lightweight, task-specific prompt modules that adapt a frozen ranking backbone to downstream tasks without full retraining.
Experiments on six public graphs and two real-world query-to-item benchmarks---Yelp2018 and MovieLens-100K---together with a controlled five-criterion graph-access study demonstrate strong ranking quality with low task-specific state overhead.
\end{abstract}

\maketitle

\section{Introduction}
\label{sec:intro}
In modern information retrieval systems, such as web search, document retrieval, and recommendation systems, ranking items according to their relevance or importance is a clearly fundamental task~\cite{page1999pagerank,kobayashi2000information,cui2020personalized,miller2004pocketlens}.
In practice, the items to be ranked can originate from diverse data sources, including graphs, documents, HTML pages, tables, and other structured or semi-structured data \cite{sen2020curious,pang2020setrank,zhuang2023rankt5,karmaker2017application}.
Among them, graph data is widely used to model social networks, biological networks, and other complex systems \cite{zhao2018ranking,he2016birank,tran2013counting,huang2014querying,10.1145/2723372.2737791,fang2020survey}. 
In this work, we focus on accurately ranking nodes in graphs.
\Cref{fig:problem_example} shows a toy example of node ranking, where given an input graph and a query (e.g., ``top-3 most-cited papers''), a ranker produces an ordered list of nodes according to a task-specific importance measure.

\begin{wrapfigure}{r}{0.52\textwidth}
    \centering
    \includegraphics[width=\linewidth]{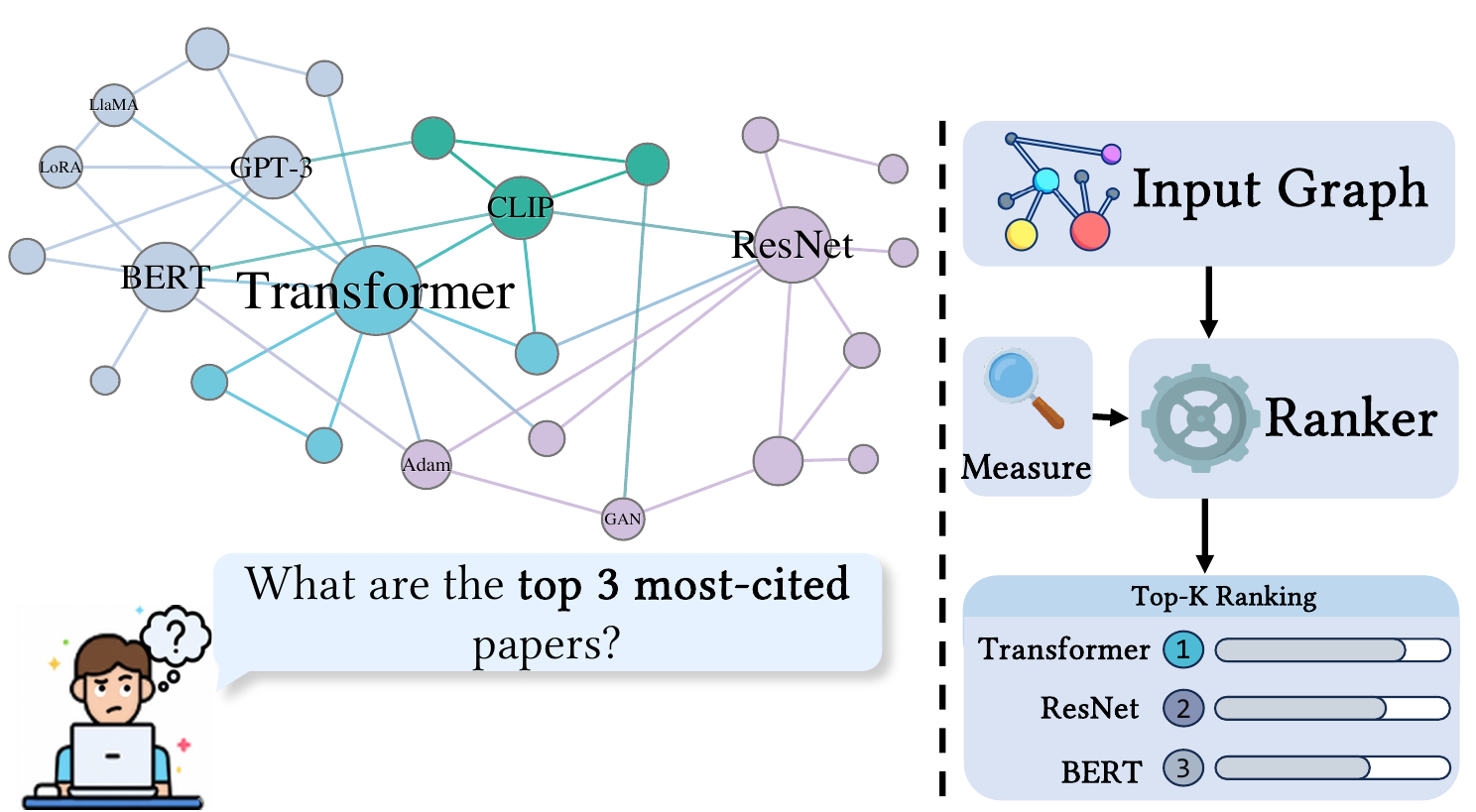}
    \caption{An illustration of the Node Ranking problem in social networks.}
    \label{fig:problem_example}
\end{wrapfigure}

In graphs, ranking nodes typically involves assigning a score to each node and ordering nodes according to these scores. In the literature, many node scoring criteria have been proposed, most of which rely on graph topology to quantify node importance.
For example, centrality-based measures, such as degree, betweenness,  local subgraph count, and eigenvector centrality, are widely used to characterize the roles of nodes in a network \cite{lee2012qube} and assess the structural quality of communities \cite{koschutzki2008centrality,lagos2024finding}.
In addition, PageRank and its variants have been extensively applied in web search \cite{page1999pagerank,haveliwala2002topic}, community search \cite{li2019optimizing,davis2006estimating}, and graph-based RAG \cite{huang2025ket}.
That is, different scoring criteria are suitable for different scenarios.

In downstream network analysis applications, efficiently computing node ranking scores is often required. To this end, numerous methods have been proposed over the past decades, which can be broadly grouped into two categories.
{\it 1) traditional methods:} This line of work typically proposes exact or approximate algorithms for computing node ranking scores~\cite{page1999pagerank,brandes2001faster,ruhnau2000eigenvector}. Both of them have inherent limitations:  the former is often computationally expensive, while the latter leads to inaccurate estimation.
{\it 2) Learning-based methods:} These methods \cite{ergashev2023learning,narang2021ranking,jafarzadeh2022learning} leverage graph neural networks (GNNs) \cite{kipf2016semi,velivckovic2017graph,hamilton2017inductive,xu2018powerful} to approximate node ranking scores with lower inference cost and higher accuracy. 
However, they suffer from a critical limitation: for each node scoring criterion (e.g., PageRank), they must first obtain node labels via human annotation, which are then used to construct the training dataset for model training.
This process is often expensive or even infeasible in practice, as it requires recollecting training data and retraining a new model for every ranking criterion, which severely limits the applicability and scalability of learning-based approaches. 
The above limitations raise a question:
\circledchar[black]{Q} {\em Can we derive unified and transferable meta-knowledge from easily accessible graph signals, enabling a single model to adapt to diverse node ranking criteria without relying on costly human annotation?}

\begin{wraptable}{R}{0.55\textwidth}
\small
\centering
\caption{Node Ranking Paradigms Comparison.}
\label{tab:node_ranking_paradigms_binary}

\resizebox{\linewidth}{!}{%
\begin{tabular}{l l c c}
\toprule
\textbf{Paradigm} & \makecell[l]{\textbf{Representative} \\ \textbf{Methods}} &
\textbf{Ranking-native} &
\textbf{Low-cost} \\
\midrule
\makecell[l]{\textbf{Traditional} \\ \textbf{Node Ranking}} &
\makecell[l]{PageRank\cite{page1999pagerank} \\Betweenness\cite{freeman1977set}} &
\cmark & \xmark \\
\addlinespace

\textbf{Vanilla GNN} &\makecell[l]{
GCN\cite{kipf2016semi} \\ GraphSage\cite{hamilton2017inductive}
}&
\cmark & \xmark \\
\addlinespace

\textbf{Finetuning} &
\makecell[l]{SimGRACE\cite{xia2022simgrace}\\ GraphMAE\cite{hou2022graphmae}} &
\xmark & \cmark \\
\addlinespace

\textbf{Prompt Tuning} &
\makecell[l]{GraphPrompt\cite{liu2023graphprompt}\\ All-in-one\cite{sun2023all}} &
\xmark & \cmark \\
\midrule

\makecell[c]{\textbf{Ranking-native}\\ \textbf{Prompt Tuning}} &
\makecell[l]{\textbf{\Method}} &
\cmark & \cmark \\
\bottomrule
\end{tabular}
}
\end{wraptable}

To address the above challenges, inspired by the success of pre-training in NLP field ~\cite{devlin2018bert,brown2020language}, several pioneering works ~\cite{hu2019strategies,hu2020gpt,xia2022simgrace,sun2022gppt,liu2023graphprompt,li2024adaptergnn} focus on designing effective pre-trained GNNs following the {\it ``pre-training and finetuning''} paradigm or {\it ``pre-training and prompt tuning''} paradigm to address the above challenges, 
Compared with ``pre-training and finetuning'' paradigm, prompt tuning based method enables the pre-trained model remain frozen, and only the input data is modified via prompts to better match the knowledge already encoded in the pre-trained model. 
This makes prompt tuning more efficient, especially for large pre-trained models, as fewer parameters are changed.
Therefore, leveraging a pre-trained graph model with prompt tuning modules for different downstream tasks could be a promising approach to address our earlier question \circledchar[black]{Q}.

{\bf Challenges.} 
Pioneering graph prompt tuning works~\cite{sun2022gppt,liu2023graphprompt} primarily target discrete classification tasks. This focus is incompatible with the continuous numerical space of our ranking objective, preventing the effective transfer of pre-trained knowledge \cite{liu2023pre} and resulting in sub-optimal performance.
To this end, to design a ranking-native framework based on the pre-training and prompt tuning paradigm, we need to address the following challenges:
\circledchar[gray]{1} {\em What task(s) should we choose as the pre-training task?} 
The pre-training task should be consistent with the downstream tasks; that is, the output space of the pre-training task should be numerical. 
Moreover, the ground-truth labels should be easy to obtain to support large-scale pre-training.
\circledchar[gray]{2} {\em How should we design the pre-training model?} 
The pre-training GNN model should capture the key information needed by diverse downstream applications while also avoiding the over-smoothing issue.
\circledchar[gray]{3} {\em How should we design the prompt tuning module for different downstream tasks?}
Different tasks may take different inputs. 
More specifically, local subgraph counting takes both a data graph and a pattern graph as input and counts how many times each vertex participates in the given pattern, whereas centrality-based prediction only requires the data graph as input.
The design of the prompt tuning module should account for both the commonalities and differences among these tasks.

\textbf{Our Solution}. Addressing these challenges, we propose {\Method}\footnote{\url{https://github.com/banrichard/PreGress}}, a ranking-native pre-training and prompting framework for learning transferable relevance scoring functions over graphs.

\underline{Pre-training Task Selection} (Challenge \circledchar[gray]{1}): We adopt {\em degree centrality prediction} as a pre-training task because it provides continuous output values and uses easily obtainable ground-truth labels—the degrees of nodes. 
This aligns with our continuous-output downstream rnaking tasks and supports large-scale pre-training. 
Additionally, we employ {\em node attribute reconstruction} as an auxiliary task to capture intrinsic relationships among node properties.

\underline{Pre-training Model Design} (Challenge \circledchar[gray]{2}): We divide the data graph into multiple {\em ego networks} and use a {\em subgraph neural network} as the pre-training model. 
This focuses node representations on local information essential for predicting node-centric graph properties and mitigates over-smoothing issues common in vanilla GNNs with many layers, since message passing occurs within confined subgraphs.

\underline{Prompt Tuning Module Design} (Challenge \circledchar[gray]{3}): We design {\em task-specific prompt strategies} to align downstream tasks with the pre-trained model's knowledge. 
For betweenness centrality tasks, we use a simple learnable vector as the prompt function to adjust graph representations. 
For local subgraph counting tasks, we incorporate a graph neural network to learn prompts based on pattern graphs. 
By tuning only the prompt functions and a unified projection head --- while keeping the pre-trained model's parameters frozen --- we enhance efficiency and reduce computational cost.

With these designs, {\Method} can be pre-trained once and then flexibly adapted to a wide range of downstream node-ranking tasks via prompt tuning.
As shown in Table~\ref{tab:node_ranking_paradigms_binary}, a notable feature of {\Method} is that it is inherently ranking-native and supports low-cost transfer to diverse downstream tasks through lightweight prompt tuning, rather than retraining the model from scratch.

We summarize our main contributions as follows:
\begin{itemize}
    \item We introduce the first \emph{ranking-native} pre-training and prompting framework for node ranking over graphs.
    
    \item We propose a \emph{multi-task pre-training strategy} to align graph pre-training objectives with the requirements of fine-grained ranking and retrieval.
    
    \item We design \emph{task-specific prompt modules} that adapt a shared, frozen ranking backbone to heterogeneous node ranking scenarios.
    
    \item Through extensive experiments on six public datasets, we demonstrate that {\Method} achieves \emph{strong and consistent ranking performance} while substantially reducing computational cost.
\end{itemize}

\section{Preliminary}
\label{sec:preliminary}
\subsection{Problem Definition}
Before formally defining the node ranking problem on graphs, we first review some basic graph concepts. 
    
A graph \(\mathcal{G} \) is denoted by a triplet \(\mathcal{G} = (V,E,\mathbf{X})\). \(V\) denotes the node set, \(E\) denotes the edge set, and \(\mathbf{X} \in \mathbb{R}^{N \times d}\) denotes the node attribute matrix that contains node properties, where $N \in \mathbb{N}_+$ represents the number of nodes in $V$ and $d\in \mathbb{N}_+$ denotes the feature dimensions of nodes, respectively. 

In this paper, the definition of node ranking task is:
\begin{problem}[Node Ranking]
\label{prb:grr}
Given a graph \(\mathcal{G}\) and a measure $m$ that specifies a ranking criterion over nodes, the goal is to learn a relevance scoring function $s_{\theta}:(v\mid m,\mathcal{G})\mapsto \mathbb{R}$,
and output a ranked list $\pi_m$ over $V$ by sorting nodes in descending order of $s_{\theta}(v\mid m,\mathcal{G})$.
\end{problem}
Here, measure $m$ is a ranking criterion which assigns each node $v$ an importance score under graph $\mathcal{G}$. For a given measure $m$, we denote by $y_m(v)\in\mathbb{R}$ the real-valued score assigned to node $v$ under $m$ on $\mathcal{G}$, which naturally induces a preference relation on $V$: for any $v_i,v_j\in V$, we write $v_i \succ_m v_j$ if and only if $y_m(v_i) > y_m(v_j)$.

We select \textbf{betweenness centrality}, and \textbf{local subgraph counting} as our downstream measures.
The definition of betweenness centrality is:
\begin{definition}[Betweenness centrality \cite{brandes2001faster}]
\label{def:between}
    Given a graph \( \mathcal{G} = (V, E, \mathbf{X}) \), the betweenness centrality \( B(i) \) of node \( i \) is defined as:
\begin{equation}
    B(i) = \sum_{s\neq i \neq t \in V} \frac{\sigma_{s,t}(i)}{\sigma_{s,t}},
\end{equation}
where \( \sigma_{s,t} \) denotes the number of shortest paths from $s$ to $t$, and \( \sigma_{s,t}(i) \) denotes the number of shortest paths that $i \in V$ lies on.
\end{definition}

The definition of local subgraph counting is:
\begin{definition}[Local subgraph counting \cite{li2024fast}]
\label{def:lsc}
    Given a graph $\mathcal{G} = (V,E,\mathbf{X})$ and a pattern graph $\mathcal{M} = (V_\mathcal{M}, E_\mathcal{M}, X_\mathcal{M})$, for node $v \in V$, local subgraph counting is to find the count of the subgraphs $\mathcal{G}[V_{\mathcal{S}}]$ of $\mathcal{G}$ that are isomorphic to $\mathcal{M}$ and contains node $v$, i.e., $\mathcal{M}\sim \mathcal{G}[V_{\mathcal{S}}]$ and $v\in V_{\mathcal{S}}$.
\end{definition}
\subsection{Graph Neural Network}
Graph Neural Network (GNN)~\cite{kipf2016semi} is a special form of neural network that is tailored to graph-structured data. 
%
%
Formally, the representation of $\mathbf{x}^{(k)}_i$ at GNN's $k$-th layer can be expressed as:
\begin{gather}
    \mathbf{x}^{(k)}_{i} = \zeta^{(k)}(\mathbf{x}^{(k-1)}_i,\mathbf{m}_i^{(k)}),\label{eq:gnn}\\
    \mathbf{m}_i^{(k)} = Agg_{j \in \mathcal{N}(i)}(\eta^{(k)}(\mathbf{x}^{(k-1)}_i,\mathbf{x}^{(k-1)}_j)),
    \label{eq:message}
\end{gather}
Here, \(\zeta(\cdot)\) denotes the update function and $\eta^{(k)}(\cdot)$ denotes the message function. 
%
%
The message $\mathbf{m}_i^{(k)}$ at $k$-th layer is calculated by Eq. \ref{eq:message}, where  $Agg(\cdot)$ denotes the permutation invariant aggregation function.
%
%
%
%

\section{Our Proposed Approach: PreGress}
\label{sec:method}
\begin{figure*}
    \centering
    \includegraphics[width=\textwidth]{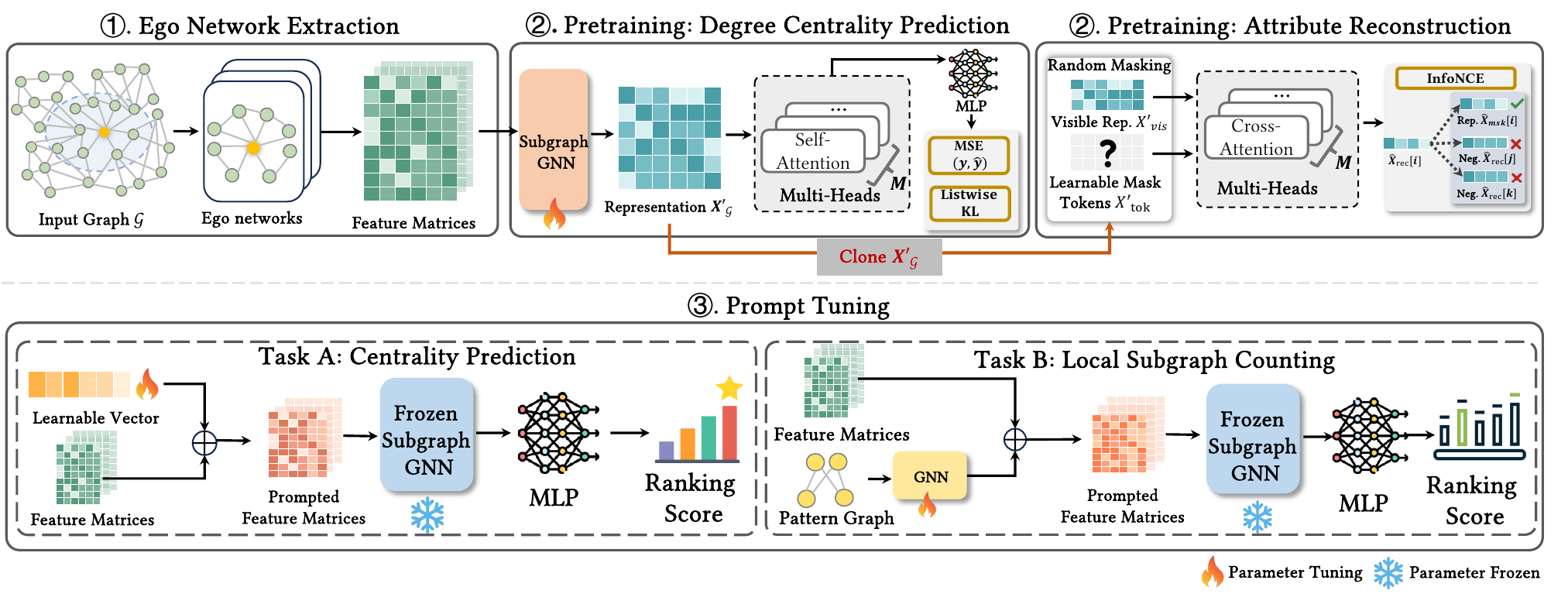}
    \caption{Overall pipeline of the proposed PreGress.}
    \label{fig:pretrain}
\end{figure*}
In this section, we present a ranking-native pre-training and prompt tuning approach, {\Method}, to rank the nodes on the given graph $\mathcal{G}$ with the specific measure $m$ efficiently; the overall framework is shown in \Cref{fig:pretrain}. {\Method} contains three stages: \circledchar[black]{1} \emph{ego network extraction} , \circledchar[black]{2} \emph{multi-task pre-training} , and \circledchar[black]{3} \emph{task-specific prompt tuning} .

The \circledchar[black]{1} \emph{ego network extraction} stage involves ego network extraction for given graph $\mathcal{G}$ and enhance graph representation with subgraph neural network for accurate rank estimation (\Cref{sec:preprocess});  \circledchar[black]{2} \emph{multi-task pre-training} stage introduces two tasks as our pre-training task to align graph pre-training objectives with the requirements of fine-grained
ranking tasks (\Cref{sec:multi_task_pretraining});  \circledchar[black]{3} \emph{task-specific prompt tuning} stage manipulates the input downstream graph with lightweight prompt according to the specific measure $m$, adapting the frozen backbone to various ranking scenarios (\Cref{sec:prompt_tuning}).

\subsection{Ego Network Extraction}
\label{sec:preprocess}
Considering Problem \ref{prb:grr}, the local structure of a node plays a crucial role in various graph-related tasks.
%
%
In vanilla GNNs, 
the size of the receptive field is determined by the number of GNN layers  $l$. 
However, the use of large $l$ in these GNN models often leads to the over-smoothing problem \cite{rusch2023survey}. 
This occurs when increasing the number of layers and causes the node embeddings to become indistinguishable. 
To mitigate this issue, we propose to use subgraph neural networks \cite{alsentzer2020subgraph,zhang2021nested,zeng2023substructure} in the pre-training phase.

%
Unlike vanilla GNNs, subgraph neural networks process multiple subgraphs induced from the original graph. 
These subgraphs are generated based on the local neighborhood structure around each node.
In our model, the node representation for each node in the subgraph is first obtained through the message-passing mechanism. 
Through the extra pooling layer, the node representations are aggregated and will be used as the central node's representation. 
Therefore, the receptive field of each node is constrained to its $l$-hop neighbors {\bf within the subgraph}, rather than its $l$-hop neighbors in the original graph.
This ensures that central node representations remain distinguishable, thus addressing the over-smoothing issue while preserving local structural and attribute information. 
For instance, the subgraph $\mathcal{G}_{sub}[i]$ is an ego network centered at node $v_i$, and the representation of center node $v_i$ is calculated by:
\begin{equation}
\label{eq:readout}
     \mathbf{x}^i_{\mathcal{G}} = readout(\{\mathbf{x}_j|v_j \in V_{sub}[v_i]\}) 
\end{equation}
Here, \(\mathbf{x}_j\) denotes the representation of node $v_j$ within the subgraph, \(readout(\cdot)\) denotes a permutation invariant pooling (readout) function 
and \(V_{sub}[v_i]\) denotes the node set in induced subgraph centered at node $v_i$. 
The final representation \(\mathbf{X}_\mathcal{G}\) is obtained by concatenating all the subgraph representation \(\mathbf{x}^i_{\mathcal{G}} \), which can be expressed as:
\begin{equation}
    \label{eq:concatgraph}
    \mathbf{X}_{\mathcal{G}} = [\mathbf{x}^0_{\mathcal{G}}||\mathbf{x}^1_{\mathcal{G}}||\dots||\mathbf{x}^{|V|}_{\mathcal{G}}]
\end{equation}

In this paper, we leverage a straightforward but effective method $k$-hop ego-network extraction to transform the original data graph into multiple subgraphs. The formal definition of $k$-hop ego-network is shown in~ \Cref{def:khop}.
\begin{definition}[$k$-hop ego network \cite{freeman1982centered}]
\label{def:khop}
Given a graph $\mathcal{G}$, the $k$-hop ego network $\mathcal{G}^k_{ego}[v]$ = $(V^k_{ego}[v],E^k_{ego}[v],\mathbf{X}^k_{ego}[v])$ centered at node $v$ is a subgraph of $\mathcal{G}$ whose node set $V^k_{ego}[v]$ includes $v$ and all nodes reachable from $v$ within $k$ hops in $\mathcal{G}$, and edge set $E^k_{ego}[v]$ includes all edges between nodes in $V^k_{ego}[v]$.
\end{definition}
    
Based on Definition \ref{def:khop}, we can easily extract $|V|$ $k$-hop ego networks and pass them to the subgraph neural network for pre-training and different downstream tasks. 
%
Here, the hop number $k$ is an important parameter that controls the structural information captured within the subgraph. Through the extracted subgraphs, we focus on local structure and align with down-stream node-level tasks.
%

\subsection{Multi-task Pre-training}
\label{sec:multi_task_pretraining}
The target of the pre-training stage is to ensure the pre-trained model can effectively capture both the structural and attribute-based information of the graph and support diverse downstream applications.
In order to achieve this, we adopt a \textbf{multi-task pre-training strategy} consisting of two key tasks: \textbf{Degree centrality prediction task}, which focuses on learning structural information, and \textbf{Attribute reconstruction task}, which ensures that the model
captures the attribute information associated with each node.
%
%
The ground-truth values for both tasks---degrees for the degree centrality prediction and node features for the attribute reconstruction task---are readily available in the input graph. This makes the training process computationally efficient since no additional external data or costly labels are required.

Before applying the pre-trained subgraph neural network to downstream tasks, we first train the subgraph neural network through a degree centrality prediction task. 
Compared with other pre-training methods that leverages supervised learning, the importance of each node is easy to obtain with heuristic methods \cite{page1999pagerank,freeman1977set}. 
The main goal of pre-training is to leverage readily available information to capture the inherent properties of graphs. 
The pre-trained models can then be used as initial models for a wide range of downstream tasks.

The degree centrality of each node is easy to obtain with traditional methods leveraging the topological information \cite{golbeck2015introduction}. 
However, predicting the importance only through topological information does not leverages the node properties, which may lead to an inaccurate prediction.
Inspired by the success of BERT in masked language modeling \cite{devlin2018bert} and Masked Autoencoder in masked image modeling \cite{he2022masked}, we select the attribute reconstruction task as an auxiliary task to augment the centrality prediction task.
In this way, the pre-trained model can exploit deeper relationship between node importance and the intrinsic properties. 
In the pre-training phase, the ego networks are fed to the subgraph neural network to learn the input graph representation \(\mathbf{X}' \in \mathbb{R}^{N \times d }\).
After we get the representation, the representation will be cloned and processed by two prediction modules separately: one self-attention prediction module to predict the degree centrality, and another with cross-attention to reconstruct the node attributes.

The main target for degree centrality prediction is to identify each node's ``importance'' regarding its local topological information. 
Therefore, an effective strategy is required to capture the relationships between nodes through their representations. 
Self-attention mechanisms are well-suited for this purpose, as they enable models to focus on relevant aspects of the input while disregarding irrelevant parts, akin to the human cognitive attention process.
One intuitive assumption is that a node's neighbors should have higher importance weights compared to nodes that are not directly connected. 
The node representations encapsulate the topological information of both the input graph and each ego network. 
Using this information, the mult-head self-attention mechanism can discern the relative importance of other nodes in relation to the target node.
From another perspective, this also optimizes the subgraph neural network to provide high-quality node representation through backpropagation.

In the self-attention mechanism, the representation matrix is transformed into three matrices: Key $\mathbf{K} = \mathbf{X}'\mathbf{W}_K$, Query $\mathbf{Q} =\mathbf{X}'\mathbf{W}_Q$, and Value $\mathbf{V} = \mathbf{X}'\mathbf{W}_V$, where $\mathbf{W}_Q \in \mathbb{R}^{d \times d_q}$,$\mathbf{W}_K \in \mathbb{R}^{d \times d_k}$ and $\mathbf{W}_V \in \mathbb{R}^{d \times d_k}$ are learnable matrices for the query, key, and value projections, respectively.
%
%
The whole self-attention process is:
\begin{equation}
    \mathbf{\hat{Y}} = \text{softmax}\left(\frac{\mathbf{Q}\mathbf{K}^T}{\sqrt{d_k}}\right)\mathbf{V}.
\end{equation}
%
%
%
The centrality can be aggregated from the node representations, with the following MLP.

For node attribute reconstruction, node representations are randomly masked with a ratio \(\sigma\), and the representations are divided into \(\mathbf{X}_{vis}' \in \mathbb{R}^{N_{vis}\times d}\) and \(\mathbf{X}_{msk}' \in \mathbb{R}^{N_{msk}\times d}\), where \(N_{vis} = N \times (1-\sigma)\) and \(N_{msk} = N \times \sigma\).
The goal is to infer the masked node attributes \(\mathbf{X}_{msk}'\) from learnable mask tokens \(\mathbf{X}_{tok}'  \in \mathbb{R}^{N_{msk}\times d}\) via the visible attributes \(\mathbf{X}_{vis}'\) and the underlying graph structure.
Each mask token represents a masked node representation, and we apply positional encoding ~\cite{ying2021transformers} for them to indicate their topological information.
In this scenario, multi-head cross-attention is more suitable than multi-head self-attention, as it better captures the relationships between the visible representation \(\mathbf{X}_{vis}'\) and learnable masked token \(\mathbf{X}_{tok}'\) via the encoded topological information.
In cross-attention, the visible representations \(\mathbf{X}_{vis}'\) serve as the key and value, and masked tokens are used as the query.
%
%
The cross-attention process can be expressed as:
\begin{equation}
    \hat{\mathbf{X}}'_{msk} = \text{softmax} \left(\frac{(\mathbf{X}_{tok}'\mathbf{W}_Q)(\mathbf{X}_{vis}'\mathbf{W}_k)^T}{\sqrt{d_k}}\right)(\mathbf{X}_{vis}'\mathbf{W}_V).
\end{equation}

In cross-attention, the attention coefficient reflects the similarity between masked attributes and visible attributes. 

With optimization, the masked token will learn the attribute similarity from the visible representation, and predict the masked representation as \(\hat{\mathbf{X}}_{msk}'\).

The visible attributes are encoded through the subgraph neural network, which contains local topological information and attribute information.

If the subgraph neural network can output a high-quality representation, the masked attribute is more likely to be reconstructed with high accuracy via the cross-attention mechanism.

Our pre-training loss combines (i) degree centrality prediction and (ii) masked attribute reconstruction.
For (i), we treat degree centrality as an intrinsic pseudo graded relevance signal and primarily perform pointwise score regression to calibrate the predicted importance values. Since node ranking ultimately depends on relative ordering, we further add a lightweight listwise distribution-matching regularizer within each mini-batch candidate set to encourage order consistency:
\begin{equation}
\label{eq:imp_combo}
\mathcal{L}_{imp}=(1-\lambda)\mathcal{L}_{mse}+\lambda\mathcal{L}_{list},
\quad
\mathcal{L}_{mse}=\frac{1}{n}\sum_{i=1}^{n}(y_i-\hat{y}_i)^2,
\end{equation}
\begin{equation}
\label{eq:list_loss}
\begin{aligned}
\mathcal{L}_{list} &= \mathrm{KL}(\mathbf{q}\|\mathbf{p}),\\
p_i &= \frac{\exp(\hat{y}_i/\tau)}{\sum_{j=1}^{N}\exp(\hat{y}_j/\tau)},\\
q_i &= \frac{\exp(y_i/\tau_y)}{\sum_{j=1}^{N}\exp(y_j/\tau_y)}.
\end{aligned}
\end{equation}
Here, $\mathcal{L}_{mse}$ stabilizes score calibration, while $\mathcal{L}_{list}$ explicitly enforces order consistency within each candidate list, making the objective ranking-oriented.

For (ii), we optimize an InfoNCE objective that matches each prediction $\hat{\mathbf{x}}'_i=\hat{\mathbf{X}}'_{\mathrm{msk}}[i]$
to its ground-truth embedding $\mathbf{x}'_i=\mathbf{X}'_{\mathrm{msk}}[i]$ against negatives from the visible set:
\begin{equation}
\label{eq:rec_infonce}
\mathcal{L}_{rec} = -\sum_{i=1}^{N_{\mathrm{msk}}}
\log
\frac{\exp(\mathrm{sim}(\hat{\mathbf{x}}'_i,\mathbf{x}'_i)/\tau_r)}
{\sum_{j=1}^{N_{\mathrm{vis}}}\exp(\mathrm{sim}(\hat{\mathbf{x}}'_i,\mathbf{X}'_{\mathrm{vis}}[j])/\tau_r)},
\end{equation}
where $\mathrm{sim}(\cdot,\cdot)$ is cosine similarity and $\tau_r$ is a temperature.

The overall pre-training loss function combines both losses:
\begin{equation}
\label{eq:preobj}
    \mathcal{L}_{pre} = \alpha\mathcal{L}_{rec} + (1-\alpha)\mathcal{L}_{imp}
\end{equation}
Through the designed pre-training architecture, the model can incorporate both topological information and node properties to predict node importance and learn the intrinsic relationships between nodes. Essentially, the model learns the influence of each node on others.  

\subsection{Task-specific Prompt Tuning}
\label{sec:prompt_tuning}

For \textbf{centrality prediction tasks}, one drawback of traditional methods \cite{page1999pagerank,freeman1977set} is the efficiency, especially on large graphs. 
Addressing the efficiency problem, we design a learnable vector $\mathbf{v} \in \mathbb{R}^{1\times d}$ as the prompt function to transform the input graph into a prompted graph. 
%
%
%
The learnable vector will be added to the down-stream graph's node attributes $X_{down}$, and the prompted graph attributes $X^p_{down}$ can be expressed as:
\begin{equation}
\label{eq:imprompt}
    \mathbf{X}^p_{down} = \mathbf{X}_{down} + \mathbf{1}^T\mathbf{v},
\end{equation}
where \(\mathbf{1}^T \in \mathbb{R}^{|V|\times 1}\) denotes a column vector whose elements are 1.
The prompted attributes $\mathbf{X}^p_{down}$ will be passed to the pre-trained subgraph neural network and downstream projection head to make the prediction. The loss function for this task is MSE loss:
\begin{equation}
\label{eq:dwnimploss}
    \mathcal{L}_{imp} = MSE(\hat{y}_{down},y_{down}).
\end{equation}

For the \textbf{local subgraph counting task}, it is more difficult than centrality prediction tasks. When working with the pattern graph, we need to take into account both its property information and topological structure.
%
%
The learnable vector is not sufficient to capture both the property information and topological information in the pattern graph.
%
%
It is intuitively assumed that the induced subgraph is likely to match the pattern graph if their features are highly similar. 
Addressing the aforementioned limitations, we apply a lightweight GNN to capture both the property information and topological information in the pattern graph.
The learned pattern graph representation \(\mathbf{v}_p\) will be used as the prompt vector and add to the down-stream graph's node attributes $X_{down}$:
\begin{equation}
    \label{eq:downlsc}
    \mathbf{X}^p_{down} =  \mathbf{X}_{down} + \mathbf{1}^T\mathbf{v}_p,
\end{equation}
where \(\mathbf{v}_p = Agg({\mathbf{x}_{i'}|x_{i'} \in \mathbf{X}_p})\) is the prompt vector learned from the lightweight GNN model trained on pattern graphs.

The above solution only solves the information capture problem, and the relationship between the pattern graph and ego network is not explored. Here, we apply soft canonical correlation analysis (Soft CCA) \cite{chang2018scalable} to exploit the similarity between the two graphs' representation. Soft CCA is applied as a regularization term in down-stream tuning, which can be expressed as:
\begin{equation}
\begin{aligned}
    \mathcal{L}_{CCA} &= \sum_{i=0}^N\mathcal{L}_{dist}(\mathbf{X}^p_{down}(i),\mathbf{v}_p)\\
    &+ \lambda (\mathcal{L}_{dl}(\mathbf{X}^p_{down}) + \mathcal{L}_{dl}(\mathbf{v}_p)),
\end{aligned}
\end{equation}
where \(\mathcal{L}_{dist}(\cdot,\cdot)\) denotes the correlation between the representation of the pattern graph and extracted ego network, which is calculated by:
\begin{equation}
    \mathcal{L}_{dist}(X^p_{down}(i),\mathbf{v}_p) = \frac{1}{2}\Vert \mathbf{X}^p_{down}(i)- \mathbf{v}_p\Vert^2_F.
\end{equation}
\(\mathcal{L}_{dl}(\cdot)\) denotes the decorrelation loss, which can be expressed as:
\begin{equation}
    \mathcal{L}_{dl}(\mathbf{U}) = \Vert \mathbf{U}^T\mathbf{U} - \mathbf{I} \Vert^2_F,
\end{equation}
where \(\mathbf{U}\) represents the normalized representation matrix and \(\mathbf{I}\) is the identity matrix. 
By minimizing the CCA loss, the prompt vector is adjusted to capture the similarity between the pattern graph and the ego networks.

The overall loss function of the local subgraph counting task can be expressed as:
\begin{equation}
\label{eq:lscobj}
        \mathcal{L}_{lsc} = MSE(\hat{y}_{down},y_{down}) + \beta\mathcal{L}_{CCA}.
\end{equation}

\subsection{Optimization Procedure}
\Cref{alg:pregress_training} summarizes how the two stages share the SNN backbone. Pre-training first optimizes the backbone and the two prediction modules on topology-derived degree targets and masked attributes. Downstream adaptation then freezes the backbone and updates only the task-specific prompt and prediction head. This separation ensures that all downstream measures reuse the same pre-trained representation rather than obtaining task-specific copies of the backbone.
\begin{algorithm}[H]
\caption{Pre-training and task-specific prompt adaptation}
\label{alg:pregress_training}
\begin{algorithmic}[1]
\STATE \textbf{Input:} Pre-training graph \(\mathcal{G}\), downstream graph \(\mathcal{G}_{down}\), optional pattern graph \(\mathcal{G}_p\), SNN \(f_s\), loss weights \(\alpha,\beta\)
\STATE Extract ego networks \(\mathcal{S}_{ego}\) from \(\mathcal{G}\)
\FOR{each pre-training mini-batch \(\mathcal{S}_b\subset\mathcal{S}_{ego}\)}
    \STATE Encode all rooted subgraphs with \(f_s\) and sum pooling
    \STATE Predict degree centrality with self-attention
    \STATE Reconstruct masked attributes with cross-attention
    \STATE Update \(f_s\) and the prediction modules using \(\mathcal{L}_{pre}\) in \Cref{eq:preobj}
\ENDFOR
\STATE Freeze the parameters of \(f_s\)
\STATE Extract downstream ego networks from \(\mathcal{G}_{down}\)
\IF{the downstream task is centrality prediction}
    \STATE Add a learnable vector prompt using \Cref{eq:imprompt}
    \STATE Update the prompt and prediction head using \(\mathcal{L}_{imp}\) in \Cref{eq:dwnimploss}
\ELSE
    \STATE Encode \(\mathcal{G}_p\) with the prompt GNN to obtain \(\mathbf{v}_p\)
    \STATE Construct prompted attributes using \Cref{eq:downlsc}
    \STATE Update the prompt GNN and prediction head using \(\mathcal{L}_{lsc}\) in \Cref{eq:lscobj}
\ENDIF
\STATE \textbf{Output:} Frozen shared backbone and task-specific prompt module
\end{algorithmic}
\end{algorithm}


\section{Theoretical Insights}
\label{sec:analysis}
\subsection{Prompting as Input-Space Adaptation for Ranking}
Learnable prompt vectors have been shown to be an effective way to adapt pre-trained models by transforming inputs in the feature space~\cite{fang2024universal}.
In this work, we go one step further and consider computed prompts: instead of directly optimizing a prompt vector, we generate the prompt via a lightweight GNN from a pattern graph, which enables structure-aware adaptation.
Below we provide a theoretical insight showing that such GNN-generated prompting can realize function adaptations that are not easily achievable by fine-tuning constrained around a fixed pre-trained solution.
\begin{definition}[Non-degeneracy condition~\cite{bishop2006pattern}]
\label{def:nondegeneracy}
For each graph, the diffusion operator is induced by a connected graph, the feature matrix has full column rank, and the target scores are distinct. These conditions exclude collinear feature bases and collapsed targets for which neither prompting nor fine-tuning can distinguish all training instances.
\end{definition}
\begin{lemma}[Vector-prompt effectiveness~\cite{fang2024universal}]
\label{lemma:promptvec_main}
For a non-degenerate collection of graphs and a fixed pre-trained GNN, there exists a target assignment for which optimizing a free input prompt vector together with a linear prediction head attains strictly lower empirical squared loss than optimizing a prediction head while constraining the backbone to the fine-tuning function class around its pre-trained solution.
\end{lemma}
\begin{theorem}
\label{theorem:promptgnn}
    For a series of graphs  \(\mathcal{D} = \{(\mathcal{G}_1,y_1),\dots,(\mathcal{G}_n,y_n)\}\), a pre-trained GNN model $f_{W_0}=(f_1,\cdots, f_n)$ with each $f_i(X_i)= \mathbbm{1} \cdot S_i X_i W_0$, and an arbitrary pattern graph \(\mathcal{G}_p\), under the non-degeneracy condition, a linear projection head $\theta$ and a GNN \(f_p\in \mathcal{F}\) for prompting, there exists \(\mathcal{Y}' = \{y'_1,\dots,y'_n\}\) for which we have
        \begin{equation}
        \begin{aligned}
            l_p &= \min_{f_p\in \mathcal{F},\theta} \sum_{i=1}^{n}(f_i(X_i + \mathbbm{1}^T\cdot f_p(X_p)) \cdot \theta - y_i)^2\\
            &< l_{ft} = \min_{\tilde{f}\in \tilde{\mathcal{F}},\theta}  \sum_{i=1}^{n}(\tilde{f}_i(X_i ) \cdot \theta - y_i)^2
        \end{aligned}
    \end{equation}
    when $y_1 = y'_1, \dots, y_n = y'_n$. Here $\mathbbm{1}$ is a vector with all items equals to 1 and with a proper size that changes from line to line ,$\tilde{\mathcal{F}}$ is defined as 
    \begin{equation*}
    \begin{aligned}
\tilde{\mathcal{F}} &= \left\{ \tilde{f} = (\tilde{f}_1, \cdots, \tilde{f}_n) : \tilde{f}_i = \mathbbm{1} \cdot S_i X_i W, 1 \leq i \leq n \right. \\
&\quad \left. \text{with} \; \forall \; (W - W_0) \right\}\label{FunctionSpaceapp}.
    \end{aligned}
    \end{equation*}
    , and $\mathcal{F}$ is defined as
    \begin{equation*}
    \begin{aligned}
    \mathcal{F}=\{ f_p:\;& f_p= \mathbbm{1}\cdot S_p X_p W_p,\\
    &{\rm with}\; \forall\; (W_p- X_p^T X_p ) \in \mathbb{R}^{d \times d} \}.
    \end{aligned}
    \end{equation*}
\end{theorem}
\begin{proof}[Proof sketch]
Let \(b_p=\mathbbm{1}S_pX_p\) denote the pooled representation of the pattern graph. Under \Cref{def:nondegeneracy}, \(b_p\neq 0\). We show that the GNN-generated prompt can realize any free prompt vector \(\mathbf{v}\) used by the vector-prompt function class. Let \(b_1=b_p/\lVert b_p\rVert_2\), extend \(b_1\) to an orthonormal basis \(\{b_1,\ldots,b_d\}\), and define \(\widetilde{W}=(b_1^T,\ldots,b_d^T)\). For \(e_1=(1,0,\ldots,0)\), choose
\begin{equation}
W_p=\frac{1}{\lVert b_p\rVert_2}
\left[b_1^T(\mathbf{v}-e_1)+\widetilde{W}\right].
\end{equation}
Then
\begin{equation}
\mathbbm{1}S_pX_pW_p=b_pW_p=\mathbf{v}.
\end{equation}
Hence, for every free vector prompt, there exists a pattern-GNN parameterization that produces the same input transformation. The GNN-generated prompt class therefore contains the vector-prompt class. Combining this construction with \Cref{lemma:promptvec_main} yields the stated strict inequality.
\end{proof}
\Cref{theorem:promptgnn} suggests that, even when the prompt is not a free vector but is generated by message passing over a pattern graph, input-space prompting can strictly enlarge the effective function class compared with fine-tuning around a fixed pre-trained solution.
This provides a theoretical justification for using GNN-generated prompts to condition downstream scoring on structural measures.
\subsection{Expressive Power of the Subgraph Backbone}
Message-passing GNNs are bounded by the discriminative power of the one-dimensional Weisfeiler--Lehman (1-WL) test when their aggregation and update functions are injective~\cite{xu2018powerful,shervashidze2011weisfeiler}.
This limitation is relevant to node ranking because structurally distinct candidates can receive indistinguishable representations even when their local roles imply different relevance scores.
Subgraph neural networks alleviate this limitation by representing every target node through its rooted ego network. Consequently, the same graph node can contribute different representations when it appears in different rooted subgraphs, whereas a conventional message-passing GNN assigns it a single global representation.

The additional expressive power of this construction has been characterized for regular graphs~\cite{zhang2021nested,chen2020can}. In particular, consider pairs of \(n\)-node \(r\)-regular graphs with \(3\leq r<\sqrt{2\log n}\). For any constant \(\epsilon>0\), an SNN with injective message-passing and readout functions can distinguish almost all such pairs using at most
\begin{equation}
\left\lceil
\left(\frac{1}{2}+\epsilon\right)
\frac{\log n}{\log(r-1-\epsilon)}
\right\rceil
\end{equation}
hops for ego-network extraction and
\(\left\lceil\epsilon\frac{\log n}{\log(r-1-\epsilon)}\right\rceil\)
message-passing layers~\cite{zhang2021nested}.
Thus, an SNN can separate graph structures that remain indistinguishable to standard message-passing GNNs without requiring ego networks that span the full graph. This result complements the ranking motivation of our design: rooted subgraphs provide node-specific structural evidence, while the bounded extraction radius keeps the representation local enough for scalable pre-training.
\subsection{Why Subgraph Neural Networks Mitigate Over-Smoothing}
We also justify the use of subgraph neural networks (SNNs) as the backbone for ranking-native pre-training.
Ranking relies on fine-grained discrimination among candidates; over-smoothing in deep GNNs collapses node representations and harms score comparability.
In contrast, SNNs restrict message passing within ego-networks, where convergence is controlled by local spectral gaps.
The formal description for \Cref{theorem:over_smooth} is:
\begin{theorem}
\label{theorem:over_smooth}
Let $G=(V, E)$ be a connected graph with normalized Laplacian $\mathbf{L}$ and second smallest eigenvalue $\lambda_2 > 0$. Let $\{G_i = (V_i, E_i)\}_{i \in V}$ be the set of $k$-hop ego-network subgraphs for an SNN, each with normalized Laplacian $\mathbf{L}_i$ and second smallest eigenvalue $\lambda_{2,i}$. Let $\lambda_{2, \text{min}} = \min_{i \in V} \{\lambda_{2,i}\}$. The convergence suppression rate to an over-smoothed state $\lambda_{2,min} \geq \lambda_2$.
\end{theorem}
This provides a theoretical motivation for SNNs in our setting: by mitigating over-smoothing, SNNs preserve representational diversity, which is crucial for stable ranking and effective prompt-based adaptation.

\subsection{Time Complexity Analysis}
In pre-processing phase, $k$-hop ego network extraction is applied, whose time complexity is \(\mathcal{O}(|V| \times |d_{max}|^k)\), where \(d_{max}\) is the maximum degree of extracted ego network.

For pre-training and prompt tuning process, we focus on their \textbf{inference time complexity}. The influence of hyper-parameters such as hidden dimensions and the number of network layers is excluded from this analysis, as they are adjustable constants and do not pertain directly to graph properties, albeit influencing practical execution time.

In the pre-training phase, the inference time complexity of each module is outlined as follows:
\begin{enumerate}
    \item Subgraph neural network: \(\mathcal{O}(|V| \times |V_{ego}| \times d_{ego})\),
with \(|V_{ego}|\) and \(d_{ego}\) indicating the maximum number
of vertices and the maximum degree in each subgraph, respectively;
    \item Attention module: \(\mathcal{O}(|V|^2)\).
\end{enumerate}
 Therefore, the comprehensive inference time complexity of the pre-training phase is \(\mathcal{O}(|V| \times |V_{ego}| \times d_{ego} + |V|^2)\).

Compared with the pre-training phase, the Transformer projection head is replaced by a multi-layer perceptron to improve efficiency. The inference time complexity of each module in prompt tuning is outlined as follows:
\begin{enumerate}
    \item Subgraph neural network: \(\mathcal{O}(|V| \times |V_{ego}| \times d_{ego})\);
    \item Prompt GNN 
    : \(\mathcal{O}(|V|_p \times d_p)\), where $|V|_p$ is the maximum number of nodes in pattern graph and $d_p$ is the maximum degree of pattern graph, respectively.
    \item Projection head: \(\mathcal{O}(|V|)\).
\end{enumerate}
The comprehensive time complexity of prompt tuning is \(\mathcal{O}(|V| \times (|V_{ego}| \times d_{ego})+ |V|_p \times d_p)\). 
With the pattern graph's size fixed, as the node number \(|V|\) increases, our approach exhibits approximately linear growth in inference time, demonstrating strong scalability.

\section{Experiments}
\label{sec:exp}
In this section, we present our experimental settings and analysis.
\subsection{Experimental Settings}
\noindent \textbf{Datasets}.
We evaluate PreGress inductively on six public graphs: soc-brightkite (BK), socfb-OR (FB), Flickr (FL), youtube (YT), web-spam (WS), indochina-2004 (ID) \cite{nr}, Yelp2018\cite{wang2019ngcf} and MovieLens-100K(MovieLens) \cite{harper2016movielens}. Their statistics are shown in \Cref{tbl:dataset}.
Following prior work, we treat each dataset as a simple undirected graph.\par
\begin{wraptable}{r}{0.64\textwidth}
\small
\caption{Statistics of the eight evaluation graphs.}
\label{tbl:dataset}
\resizebox{\linewidth}{!}{
\begin{tabular}{c|c|c|c|c|c}
   \toprule
   Dataset & Notation&Category&\( |V| \) & \( |E| \) & avg. deg. \\
   \midrule
   web-spam &WS&Email net.&$ 4.8 \times 10^3$& $3.7 \times 10^4$ &15.7\\
   indochina-2004&ID &Web net.& $1.1 \times 10^4$ & $4.7 \times 10^4$ & 8.0 \\
   soc-brightkite&BK &Social net.& $5.7\times 10^4$ & $2.1\times 10^5$& 7.0  \\
    socfb-OR &FB&Facebook net.&$6.3 \times 10^4$&$8.2\times 10^5$&25.0\\
   Flickr&FL &Image net.& $8.9 \times 10^4$ & $8.9 \times 10^5$ & 10.0 \\
   youtube &YT&Video net.& $1.1\times 10^6$ & $3.0\times 10^6$ & 5.3  \\
   Yelp2018       & YP & User--item & $7.0 \times 10^4$ & $1.2 \times 10^6$ & 34.6 \\
MovieLens-100K  & ML & User--item & $2.3 \times 10^3$ & $5.2 \times 10^4$ & 44.2 \\
   \bottomrule
\end{tabular}
}
\end{wraptable}

\noindent \textbf{Ranking measures.}
We study node ranking using two downstream measures:
betweenness centrality prediction and local subgraph counting; to provide an edge-level contrast, we further evaluate link prediction as an additional downstream task.
The first two tasks define ranked lists over nodes by sorting predicted relevance scores, while link prediction ranks candidate edges for each query node.

For local subgraph counting, we follow SCOPE \cite{li2024fast} and tune on 515 pattern graphs in total: 58 patterns with 5 nodes, 407 patterns with 6 nodes, and 50 patterns with 7 nodes. The ground truth labels for local subgraph counting is obtained by SCOPE \cite{li2024fast}.

\noindent \textbf{Baseline methods}. To align with the paradigm taxonomy in \Cref{tab:node_ranking_paradigms_binary}, we compare methods from the following categories:
\begin{enumerate*}
    \item Non-learning rankers that compute importance scores directly on graphs:
    exact centrality solver \cite{brandes2001faster} and subgraph counting engines (SCOPE \cite{li2024fast}, DISC \cite{zhang2020distributed}).
    \item Supervised rankers that require training a task-specific model for each ranking measure: GNN encoders (GCN \cite{kipf2016semi}, GAT \cite{velivckovic2017graph}, GraphSAGE \cite{hamilton2017inductive}, GIN \cite{xu2018powerful}, etc.) with regression- or ranking-style objectives.
    \item Pre-training based adaptation methods that reuse a pretrained backbone and adapt to a new measure via fine-tuning (e.g., SimGRACE \cite{xia2022simgrace}, GraphMAE \cite{hou2022graphmae}, AdapterGNN \cite{li2024adaptergnn}) or prompt tuning (e.g., GPPT \cite{sun2022gppt}, All-in-One \cite{sun2023all}, GraphPrompt \cite{liu2023graphprompt}).
\end{enumerate*}

\noindent \textbf{Evaluation Metrics.} Our primary goal is to induce correct rankings. For node-ranking tasks, we report ranking quality using \texttt{NDCG@10} and \texttt{NDCG@20}. For link prediction, we follow standard edge-ranking evaluation and report source-grouped \texttt{MRR}, \texttt{Hits@10}, and \texttt{Hits@20}. For efficiency, we report average per-query inference time.

\noindent \textbf{Implementation and hyper-parameter settings.}
We implement PreGress with PyTorch and PyTorch Geometric. Unless otherwise specified, we extract \(2\)-hop ego networks and use a five-layer SNN with sum pooling. The input feature dimension is \(64\), the SNN output dimension is \(128\), and both pre-training prediction modules use three-layer Transformers with an attention dimension of \(64\). The pre-training loss weight \(\alpha\) and mask ratio \(\sigma\) are selected from \([0.1,0.9]\).
During prompt tuning, the learnable centrality prompt has the same dimension as the SNN output. For local subgraph counting, the prompt GNN contains three layers, its hidden dimension is selected from \(\{32,64,128\}\), and its output is passed to a three-layer MLP. The CCA regularization weight \(\beta\) is selected from \([0.1,0.9]\).
Both stages use Adam~\cite{kingma2014adam} with a learning rate selected from \([10^{-5},10^{-3}]\), weight decay from \([10^{-5},10^{-2}]\), and an ExponentialLR scheduler~\cite{li2019exponential} with decay factor from \([0.5,0.9]\). We tune batch size in \(\{4,8,16,32,64\}\) and train for \(80\)--\(200\) epochs depending on the dataset. Each downstream dataset is split into \(80\%\), \(10\%\), and \(10\%\) for training, validation, and testing, respectively. Baseline hyper-parameters are tuned using their released implementations under the same data splits. DISC is executed with Spark 2.4.3 in a single-machine configuration using 32 threads.
All experiments run on a server with Ubuntu 24, four 28-core Intel Xeon Gold 6330 CPUs, eight NVIDIA RTX A5000 GPUs, and 1\,TB RAM.

\begin{table*}[htbp]
    \centering
    \caption{Comparison of ranking performance w.r.t betweenness centrality prediction across six datasets. The best results are \textbf{bolded} and the second-best results are \underline{underlined}. }
    \label{tab:betweenness_mape}
\resizebox{\textwidth}{!}{
\begin{tabular}{l|c|cc|cc|cc|cc|cc|cc}
\toprule
        \multirow{2}{*}{\textbf{Scheme}} & \multirow{2}{*}{\textbf{Method}} & 
        \multicolumn{2}{c|}{\textbf{BK}} & \multicolumn{2}{c|}{\textbf{FB}} & 
        \multicolumn{2}{c|}{\textbf{FL}} & \multicolumn{2}{c|}{\textbf{ID}} & 
        \multicolumn{2}{c|}{\textbf{WS}} & \multicolumn{2}{c}{\textbf{YT}} \\
        \cmidrule(lr){3-4} \cmidrule(lr){5-6} \cmidrule(lr){7-8} \cmidrule(lr){9-10} \cmidrule(lr){11-12} \cmidrule(lr){13-14}
        & & \textbf{NG@10} & \textbf{NG@20} & \textbf{NG@10} & \textbf{NG@20} & \textbf{NG@10} & \textbf{NG@20} & \textbf{NG@10} & \textbf{NG@20} & \textbf{NG@10} & \textbf{NG@20} & \textbf{NG@10} & \textbf{NG@20} \\
    \midrule
        \multirow{7}*{GNN} 
         & GCN         & 0.1799 & 0.1911 & 0.1296 & 0.1693 & 0.1047 & 0.1179 & 0.1153 & 0.1246 & 0.0966 & 0.1021 & 0.1530 & 0.1465 \\
         & GAT         & 0.1259 & 0.1379 & 0.1653 & 0.1591 & 0.1097 & 0.1149 & 0.1519 & 0.1685 & 0.1515 & 0.1407 & 0.1567 & 0.1499 \\
         & GraphSage   & 0.1912 & 0.1818 & 0.1727 & 0.1649 & 0.1437 & 0.1339 & 0.0940 & 0.1214 & 0.0915 & 0.0953 & 0.1398 & 0.1273 \\
         & GIN         & 0.1648 & 0.1937 & 0.1540 & \textbf{0.1868} & 0.1757 & 0.1508 & 0.1221 & 0.1044 & 0.1569 & 0.1486 & 0.1595 & 0.1228 \\
         & GraphGPS    & 0.1428 & 0.1566 & 0.0891 & 0.0770 & 0.1640 & 0.1598 & 0.1538 & 0.1428 & 0.1278 & 0.0906 & 0.1465 & 0.1401 \\
         & DeepGCN     & 0.1521 & 0.1147 & 0.1492 & 0.1332 & 0.1345 & 0.1438 & 0.1408 & 0.1060 & 0.1156 & 0.0951 & 0.1285 & 0.1224 \\
         & GENI        & 0.1382 & 0.1526 & 0.0903 & 0.0925 & 0.1414 & 0.1527 & 0.2367 & 0.2329 & 0.1397 & 0.0818 & 0.1319 & 0.1258 \\
        \hline
        \multirow{3}*{Finetuning} 
         & SimGRACE    & 0.1602 & 0.1861 & 0.1204 & 0.1252 & \underline{0.5902} & 0.5054 & 0.2790 & 0.4192 & 0.1371 & 0.1341 & \underline{0.1940} & 0.1768 \\
         & GraphMAE    & 0.3355 & \underline{0.3465} & 0.1318 & 0.1471 & 0.4990 & \textbf{0.4897} & 0.2905 & 0.3378 & 0.1322 & 0.1310 & 0.1504 & 0.1225 \\
         & AdapterGNN  & \underline{0.3398} & 0.3271 & 0.1751 & 0.1567 & 0.2671 & 0.2838 & 0.2020 & 0.2744 & 0.1313 & 0.1226 & 0.1673 & 0.1620 \\
        \hline
        \multirow{4}*{Prompt tuning} 
         & GraphPrompt & 0.3319 & 0.2765 &0.1229 & 0.1228 & 0.4545 & 0.5017 & 0.3473 & 0.3632 & \underline{0.1662} & \underline{0.1872} & 0.1356 & 0.1311 \\
         & GPPT        & 0.3003 & 0.3019 & 0.1082 & 0.1091 & 0.4909 & 0.4211 & \underline{0.4373} & \textbf{0.4831} & 0.1344 & 0.1369 & 0.1555 & 0.1487 \\
         & AllinOne    & 0.2859 & 0.2793 & \underline{0.1757} & 0.1470 & 0.4066 & 0.4915 & 0.4008 & 0.3718 & 0.1315 & 0.1314 & 0.1800 & \underline{0.1781} \\
         & PreGress    & \textbf{0.3681} & \textbf{0.3610} & \textbf{0.1788} & \underline{0.1719} & \textbf{0.6667} & \textbf{0.5737} & \textbf{0.4820} & \underline{0.4579} & \textbf{0.1916} & \textbf{0.1879} & \textbf{0.1959} & \textbf{0.1805} \\
    \bottomrule
    \end{tabular}
}
\end{table*}
\begin{table*}[htbp]
    \centering
    \caption{Comparison of ranking performance w.r.t local subgraph counting across six datasets.}
    \resizebox{\textwidth}{!}{
    \begin{tabular}{c|c|cc|cc|cc|cc|cc|cc}
    \toprule
        \multirow{2}{*}{\textbf{Category}} & \multirow{2}{*}{\textbf{Methods}} & 
        \multicolumn{2}{c|}{\textbf{BK}} & \multicolumn{2}{c|}{\textbf{FB}} & 
        \multicolumn{2}{c|}{\textbf{FL}} & \multicolumn{2}{c|}{\textbf{ID}} & 
        \multicolumn{2}{c|}{\textbf{WS}} & \multicolumn{2}{c}{\textbf{YT}} \\
        \cmidrule(lr){3-4} \cmidrule(lr){5-6} \cmidrule(lr){7-8} \cmidrule(lr){9-10} \cmidrule(lr){11-12} \cmidrule(lr){13-14}
        ~ & ~ & \textbf{NG@10} & \textbf{NG@20} & \textbf{NG@10} & \textbf{NG@20} & \textbf{NG@10} & \textbf{NG@20} & \textbf{NG@10} & \textbf{NG@20} & \textbf{NG@10} & \textbf{NG@20} & \textbf{NG@10} & \textbf{NG@20} \\
    \midrule

    \multirow{1}*{Traditional} 
        & DISC & 0.1484 & 0.1611 & 0.1437 & 0.1332 & 0.1544 & 0.1376 & 0.1712 & 0.1699 & 0.1534 & 0.1478 & 0.1597 & 0.1589 \\
    \hline
    
    \multirow{4}*{Prompt Tuning}
        & GraphPrompt & \underline{0.1838} & 0.1613 & 0.1223 & 0.1213 & \textbf{0.2172} & \textbf{0.2362} & \underline{0.1745} & \underline{0.1678} & 0.1472 & 0.1388 & 0.1284 & 0.1312 \\
        & GPPT        &0.1078 & 0.1057 & 0.1675 & 0.1596 & 0.1511 & 0.1364 & 0.1655& 0.1676 & \underline{0.1774}  & 0.1540 & 0.1658 & 0.1348 \\
        & AllinOne    & 0.1762 & \underline{0.1705} & \underline{0.1810} & \underline{0.1724} & 0.1578 & 0.1301 & 0.1539 & 0.1550 & 0.1263 & \underline{0.1575} & \underline{0.1700} &\underline{0.1672} \\
        & PreGress    & \textbf{0.1847} & \textbf{0.1914} & \textbf{0.2095} &\textbf{ 0.2022} & \underline{0.1863 }& \underline{0.1820} & \textbf{0.1978} & \textbf{0.1877} & \textbf{0.1848} & \textbf{0.1752} & \textbf{0.2275} & \textbf{0.2254} \\
    \bottomrule
    \end{tabular}
    }
    \label{tab:lscacc}
\end{table*}
\begin{table*}[htbp]
    \centering
    \caption{Comparison of link prediction performance across six datasets. We report source-grouped \texttt{MRR}, \texttt{H@10}, and \texttt{H@20}, where H@K denotes Hits@K. The best results are \textbf{bolded} and the second-best results are \underline{underlined}.}
    \label{tab:link_prediction_hits_mrr}
\resizebox{\textwidth}{!}{
\begin{tabular}{l|c|ccc|ccc|ccc|ccc|ccc|ccc}
\toprule
        \multirow{2}{*}{\textbf{Scheme}} & \multirow{2}{*}{\textbf{Method}} & 
        \multicolumn{3}{c|}{\textbf{BK}} & \multicolumn{3}{c|}{\textbf{FB}} & 
        \multicolumn{3}{c|}{\textbf{FL}} & \multicolumn{3}{c|}{\textbf{ID}} & 
        \multicolumn{3}{c|}{\textbf{WS}} & \multicolumn{3}{c}{\textbf{YT}} \\
        \cmidrule(lr){3-5} \cmidrule(lr){6-8} \cmidrule(lr){9-11} \cmidrule(lr){12-14} \cmidrule(lr){15-17} \cmidrule(lr){18-20}
        & & \textbf{MRR} & \textbf{H@10} & \textbf{H@20} & \textbf{MRR} & \textbf{H@10} & \textbf{H@20} & \textbf{MRR} & \textbf{H@10} & \textbf{H@20} & \textbf{MRR} & \textbf{H@10} & \textbf{H@20} & \textbf{MRR} & \textbf{H@10} & \textbf{H@20} & \textbf{MRR} & \textbf{H@10} & \textbf{H@20} \\
    \midrule
        \multirow{7}*{GNN} 
         & GCN         & 0.306 & 0.587 & 0.744 & 0.505 & 0.833 & 0.929 & 0.280 & 0.582 & 0.768 & 0.638 & 0.824 & 0.911 & 0.511 & 0.866 & 0.957 & 0.483 & 0.728 & 0.835 \\
         & GAT         & 0.383 & 0.679 & 0.816 & 0.508 & 0.848 & 0.935 & 0.310 & 0.589 & 0.760 & 0.696 & 0.848 & 0.930 & 0.452 & 0.794 & 0.931 & 0.400 & 0.714 & 0.829 \\
         & GraphSage   & 0.314 & 0.562 & 0.740 & 0.593 & 0.881 & 0.952 & 0.376 & 0.630 & 0.791 & 0.676 & 0.843 & 0.921 & 0.400 & 0.741 & 0.912 & 0.512 & 0.757 & 0.839 \\
         & GIN         & 0.421 & 0.635 & 0.742 & 0.498 & 0.829 & 0.923 & 0.348 & 0.597 & 0.782 & 0.295 & 0.535 & 0.737 & 0.454 & 0.820 & 0.939 & 0.429 & 0.706 & 0.830 \\
         & GraphGPS    & 0.289 & 0.537 & 0.734 & 0.533 & 0.847 & 0.933 & 0.253 & 0.516 & 0.718 & 0.667 & 0.849 & 0.925 & 0.324 & 0.649 & 0.840 & 0.463 & 0.710 & 0.819 \\
         & DeepGCN     & 0.444 & 0.741 & 0.866 & 0.562 & 0.899 & 0.958 & 0.276 & 0.546 & 0.751 & 0.660 & 0.894 & 0.954 & 0.543 & 0.892 & 0.960 & \underline{0.566} & \textbf{0.823} & \underline{0.886} \\
         & GENI        & 0.350 & 0.612 & 0.742 & 0.421 & 0.758 & 0.883 & 0.242 & 0.468 & 0.652 & 0.640 & 0.863 & 0.917 & 0.420 & 0.734 & 0.912 & 0.387 & 0.681 & 0.811 \\
        \hline
        \multirow{3}*{Finetuning} 
         & SimGRACE    & 0.338 & 0.630 & 0.777 & 0.469 & 0.825 & 0.923 & 0.322 & 0.600 & 0.750 & 0.526 & 0.821 & 0.927 & 0.339 & 0.618 & 0.789 & 0.399 & 0.674 & 0.798 \\
         & GraphMAE    & 0.320 & 0.614 & 0.765 & 0.476 & 0.831 & 0.920 & 0.332 & 0.587 & 0.744 & 0.525 & 0.829 & 0.926 & 0.496 & 0.848 & 0.935 & 0.415 & 0.677 & 0.807 \\
         & AdapterGNN  & \underline{0.532} & \underline{0.801} & \underline{0.880} & \underline{0.689} & \underline{0.930} & \textbf{0.974} & \underline{0.380} & \underline{0.677} & \underline{0.833} & \underline{0.809} & \underline{0.937} & \underline{0.966} & 0.451 & 0.776 & 0.914 & 0.561 & 0.795 & 0.874 \\
        \hline
        \multirow{4}*{Prompt tuning} 
         & GraphPrompt & 0.366 & 0.603 & 0.687 & 0.458 & 0.823 & 0.921 & 0.325 & 0.588 & 0.756 & 0.526 & 0.828 & 0.924 & 0.326 & 0.590 & 0.753 & 0.413 & 0.673 & 0.798 \\
         & GPPT        & 0.351 & 0.638 & 0.799 & 0.474 & 0.830 & 0.927 & 0.323 & 0.592 & 0.751 & 0.547 & 0.839 & 0.930 & 0.581 & 0.714 & 0.847 & 0.412 & 0.678 & 0.818 \\
         & AllinOne    & 0.343 & 0.627 & 0.775 & 0.441 & 0.817 & 0.922 & 0.327 & 0.590 & 0.749 & 0.527 & 0.831 & 0.924 & \underline{0.590} & \textbf{0.914} & \textbf{0.976} & 0.414 & 0.685 & 0.812 \\
         & PreGress    & \textbf{0.534} & \textbf{0.805} & \textbf{0.906} & \textbf{0.748} & \textbf{0.972} & \underline{0.960} & \textbf{0.479} & \textbf{0.756} & \textbf{0.891} & \textbf{0.817} & \textbf{0.948} & \textbf{0.976} & \textbf{0.603} & \underline{0.905} & \underline{0.961} & \textbf{0.576} & \underline{0.808} & \textbf{0.892} \\
    \bottomrule
    \end{tabular}
}
\end{table*}
\subsection{Accuracy Comparison}
\Cref{tab:betweenness_mape} shows that PreGress achieves the most consistent and strongest ranking performance for betweenness centrality prediction across six datasets. It achieves the best NDCG@10/20 on BK, FL, ID, WS, and YT, and remains highly competitive on FB, where it ranks second at NDCG@20 and is only slightly below the best baseline. In contrast, vanilla GNN models perform worse across datasets, suggesting that standard message passing are insufficient to conduct the node ranking task. While fine-tuning and existing prompt-tuning baselines can yield large gains on specific graphs, their improvements are less stable across datasets. Overall, the results indicate that PreGress better aligns pre-training and adaptation with the ranking objective, leading to improved cross-graph generalization and more reliable top-$k$ ordering.

For the \textbf{local subgraph counting} task, we compare our methods with the SOTA prompt tuning methods and traditional method DISC \cite{zhang2020distributed}. Ground truth labels are calculated by SCOPE \cite{li2024fast}.
Table \ref{tab:lscacc} reports NDCG@10 and NDCG@20 for local subgraph counting. Overall, PreGress achieves the most consistent and strongest ranking performance across datasets: it attains the best results on BK, FB, ID, WS, and YT for both cutoffs, and remains highly competitive on FL. Compared with the traditional baseline DISC, PreGress yields clear and consistent improvements in both NDCG@10 and NDCG@20 on all datasets, demonstrating that learned ranking signals are more effective than exact counting heuristics for top-$k$ ordering. Among prompt tuning baselines, performance is strongly dataset-dependent. GraphPrompt excels on FL and AllinOne shows moderate but stable results—whereas PreGress maintains high NDCG at both cutoffs across all graphs, indicating better robustness and cross-graph generalization for local structural ranking tasks.
%

\Cref{tab:link_prediction_hits_mrr} reports the edge-level link prediction results using standard link-prediction metrics. Although PreGress is designed for node-ranking tasks rather than edge prediction, it still achieves the best results on most metrics, including all three metrics on BK, FL, and ID, MRR/H@10 on FB, MRR on WS, and MRR/H@20 on YT. AdapterGNN remains a strong edge-level baseline and obtains the best H@20 on FB, while AllinOne and DeepGCN are competitive on WS and YT Hits metrics, respectively. These results show that the ranking-native pre-training signals used by PreGress do not sacrifice edge-level transferability. More importantly, when considered together with \Cref{tab:betweenness_mape,tab:lscacc}, they indicate that methods with isolated advantages in link prediction do not necessarily transfer to node-ranking tasks, where PreGress is consistently stronger. This supports our motivation that ranking-oriented node structural signals require task-aware pre-training and prompting beyond conventional link-prediction objectives.

%
%

Considering the two node-ranking downstream tasks, PreGress obtains decent results and shows great generalization ability to perform well on most of the datasets.
The results validate that the pre-trained model tends to perform well when the pre-training task and the downstream task are within the same ground-truth label space (answering challenge \circledchar{2} in \Cref{sec:intro}).

\subsection{Real-World Information Access via Recommendation}
\label{sec:recommendation_access}
We next instantiate graph information access as a real recommendation task: given a user/query node, the system ranks candidate item nodes and returns the top-$k$ results.
This evaluation tests whether a pretrained graph backbone can transfer beyond synthetic structural criteria while respecting the train-only interaction boundary.

\begin{table}[t]
\centering
\caption{Real-world query-to-item information access. Results are mean $\pm$ sample standard deviation over three seeds. Latency includes scoring, masking known positives, and top-20 retrieval with cached embeddings. Best and second-best retrieval results within each dataset are \textbf{bolded} and \underline{underlined}.}
\label{tab:recommendation_access}
\footnotesize
\setlength{\tabcolsep}{3.7pt}
\begin{tabular*}{\textwidth}{@{\extracolsep{\fill}}llcccc@{}}
\toprule
Dataset & Method & Recall@20 & NDCG@20 & MRR & Latency (ms) \\
\midrule
\multirow{6}{*}{Yelp2018}
& LightGCN
& $0.0452{\scriptstyle\pm0.0003}$ & $0.0371{\scriptstyle\pm0.0003}$
& $0.0890{\scriptstyle\pm0.0006}$ & $\mathbf{0.1075}{\scriptstyle\pm0.0014}$ \\
& SimGCL
& $\mathbf{0.0638}{\scriptstyle\pm0.0002}$ & $\mathbf{0.0524}{\scriptstyle\pm0.0000}$
& $\mathbf{0.1186}{\scriptstyle\pm0.0002}$ & $0.1105{\scriptstyle\pm0.0033}$ \\
& GraphMAE
& $0.0356{\scriptstyle\pm0.0007}$ & $0.0286{\scriptstyle\pm0.0005}$
& $0.0706{\scriptstyle\pm0.0013}$ & $0.1100{\scriptstyle\pm0.0019}$ \\
& AdapterGNN
& $0.0458{\scriptstyle\pm0.0002}$ & $0.0372{\scriptstyle\pm0.0002}$
& $0.0882{\scriptstyle\pm0.0004}$ & $0.1103{\scriptstyle\pm0.0024}$ \\
& GraphPrompt
& $0.0457{\scriptstyle\pm0.0000}$ & $0.0372{\scriptstyle\pm0.0000}$
& $0.0883{\scriptstyle\pm0.0000}$ & $0.1140{\scriptstyle\pm0.0022}$ \\
& PreGress-CNA
& $\underline{0.0636}{\scriptstyle\pm0.0000}$ & $\underline{0.0523}{\scriptstyle\pm0.0000}$
& $\underline{0.1185}{\scriptstyle\pm0.0000}$ & $\underline{0.1093}{\scriptstyle\pm0.0015}$ \\
\midrule
\multirow{7}{*}{MovieLens}
& Popularity
& $0.5342{\scriptstyle\pm0.0000}$ & $0.2352{\scriptstyle\pm0.0000}$
& $0.1651{\scriptstyle\pm0.0000}$ & $\mathbf{0.2359}{\scriptstyle\pm0.0109}$ \\
& MF
& $0.5266{\scriptstyle\pm0.0089}$ & $0.2154{\scriptstyle\pm0.0024}$
& $0.1418{\scriptstyle\pm0.0008}$ & $0.3852{\scriptstyle\pm0.1111}$ \\
& LightGCN
& $0.5516{\scriptstyle\pm0.0050}$ & $0.2457{\scriptstyle\pm0.0022}$
& $0.1730{\scriptstyle\pm0.0013}$ & $0.3184{\scriptstyle\pm0.0427}$ \\
& PreGress-Struct
& $0.5324{\scriptstyle\pm0.0081}$ & $0.2066{\scriptstyle\pm0.0065}$
& $0.1294{\scriptstyle\pm0.0060}$ & $0.4188{\scriptstyle\pm0.0893}$ \\
& PreGress-CNA
& $0.5480{\scriptstyle\pm0.0017}$ & $0.2463{\scriptstyle\pm0.0004}$
& $0.1748{\scriptstyle\pm0.0000}$ & $\underline{0.2762}{\scriptstyle\pm0.0412}$ \\
& Identity (no pre-training)
& $\underline{0.7300}{\scriptstyle\pm0.0125}$ & $\underline{0.3365}{\scriptstyle\pm0.0065}$
& $\underline{0.2332}{\scriptstyle\pm0.0048}$ & $0.3462{\scriptstyle\pm0.0296}$ \\
& PreGress-ID
& $\mathbf{0.7474}{\scriptstyle\pm0.0017}$ & $\mathbf{0.3501}{\scriptstyle\pm0.0032}$
& $\mathbf{0.2454}{\scriptstyle\pm0.0033}$ & $0.3409{\scriptstyle\pm0.0366}$ \\
\bottomrule
\end{tabular*}
\end{table}

On Yelp2018, the strongest native recommender, SimGCL, obtains 0.0638 Recall@20, 0.0524 NDCG@20, and 0.1186 MRR.
PreGress-CNA reaches 0.0636, 0.0523, and 0.1185, respectively; its relative gaps are below 0.30\% for both Recall@20 and NDCG@20.
This comparable retrieval quality is achieved by adapting 4,417 parameters for \(8.94\) seconds on average, versus 4.46 million parameters and \(349.42\) seconds for SimGCL---a 99.9\% reduction in trainable state and a \(39.1\times\) downstream-training speedup.
GraphMAE underperforms the native recommenders, while AdapterGNN and GraphPrompt remain near LightGCN, indicating that simply importing a generic structural objective or lightweight module is insufficient without a recommendation-aligned transfer signal.

MovieLens exposes an objective-dependent boundary.
PreGress-ID is best, achieving 0.7474 Recall@20, 0.3501 NDCG@20, and 0.2454 MRR.
In contrast, PreGress-CNA obtains 0.5480/0.2463/0.1748 and is statistically indistinguishable from LightGCN on NDCG@20 under a paired three-seed test (\(p=0.633\)).
Thus, collaborative-neighborhood alignment efficiently transfers the teacher signal but does not replace direct identity pre-training when stable user/item identities provide strong transductive supervision.
This negative boundary is consistent across history buckets and prevents us from claiming that a single alignment objective dominates every recommendation regime.

\begin{table}[t]
\centering
\caption{NDCG@20 by user history length. A dash indicates that Yelp2018 contains no evaluated user in the corresponding bucket after one training interaction per user is reserved for validation.}
\label{tab:recommendation_history}
\footnotesize
\setlength{\tabcolsep}{5pt}
\begin{tabular*}{\textwidth}{@{\extracolsep{\fill}}llcccc@{}}
\toprule
Dataset & Method & 1--10 & 11--20 & 21--50 & 51+ \\
\midrule
\multirow{3}{*}{Yelp2018}
& LightGCN & -- & 0.0319 & 0.0364 & 0.0499 \\
& SimGCL & -- & 0.0471 &  0.0513 & 0.0651 \\
& PreGress-CNA & -- & \textbf{0.0472} &\textbf{0.0516} &  \textbf{0.0653}\\
\midrule
\multirow{3}{*}{MovieLens}
& LightGCN & 0.2818 & 0.2456 & 0.2274 & 0.2507 \\
& PreGress-CNA & 0.2849 & 0.2449 & 0.2246 & 0.2542 \\
& PreGress-ID & \textbf{0.4530} & \textbf{0.3935} & \textbf{0.3699} & \textbf{0.2804} \\
\bottomrule
\end{tabular*}
\end{table}

\Cref{tab:recommendation_history} further shows that PreGress-CNA performs best in every observed Yelp2018 history bucket, while PreGress-ID is best in all four MovieLens buckets.
Together, the two datasets support a conditional information-access claim: PreGress can preserve native-recommender quality with much smaller task-specific updates, but the appropriate pre-training signal depends on whether the deployment provides reusable identity evidence or primarily neighborhood structure.

\subsection{Ablation Studies}
We also conduct ablation studies on WS and FL for local subgraph counting and betweenness centrality prediction: replacing the subgraph neural network with GCN (w/o SNN), replacing the prompting GNN with a fully connected layer (w/o prompting GNN), removing prompt tuning (w/o tuning), and removing ego-net extraction (w/o ego net).

\begin{table}[t]
    \centering
    \footnotesize
    \caption{Ablation performance for betweenness centrality prediction and local subgraph counting on WS and FL. Results are NDCG@10/20; the best and second-best results within each task are \textbf{bolded} and \underline{underlined}.}
    \label{tab:abl}
    \setlength{\tabcolsep}{4.5pt}
    \begin{tabular*}{\textwidth}{@{\extracolsep{\fill}}llcccc@{}}
    \toprule
        \multirow{2}{*}{Task} & \multirow{2}{*}{Strategy} &
        \multicolumn{2}{c}{WS} & \multicolumn{2}{c}{FL} \\
        \cmidrule(lr){3-4}\cmidrule(lr){5-6}
        & & NDCG@10 & NDCG@20 & NDCG@10 & NDCG@20 \\
    \midrule
        \multirow{4}{*}{Betweenness}
        & w/o tuning  & \underline{0.1363} & \underline{0.1295} & \underline{0.4101} & \underline{0.4775} \\
        & w/o SNN     & 0.1077 & 0.1051 & 0.2428 & 0.2066 \\
        & w/o ego net & 0.0859 & 0.0718 & 0.2881 & 0.2793 \\
        & PreGress    & \textbf{0.1916} & \textbf{0.1879} & \textbf{0.6667} & \textbf{0.5737} \\
    \midrule
        \multirow{5}{*}{\makecell[l]{Local subgraph\\counting}}
        & w/o tuning        & 0.1487 & 0.1470 & 0.1038 & 0.0914 \\
        & w/o SNN           & \underline{0.1514} & \underline{0.1483} & \underline{0.1724} & \underline{0.1646} \\
        & w/o prompting GNN & 0.1133 & 0.1041 & 0.1242 & 0.1093 \\
        & w/o ego net       & 0.1323 & 0.1307 & 0.1502 & 0.1500 \\
        & PreGress          & \textbf{0.1848} & \textbf{0.1752} & \textbf{0.1863} & \textbf{0.1820} \\
    \bottomrule
    \end{tabular*}
\end{table}

\Cref{tab:abl} reports the results in terms of NDCG@10 and NDCG@20. We have three main findings (addressing Challenge \circledchar{3} in \Cref{sec:intro}):
\begin{enumerate*}
    \item \textbf{Prompt tuning is essential}: w/o tuning consistently yields much lower NDCG@10/20, indicating that task-aware adaptation is crucial for ranking.
    \item \textbf{Subgraph-aware pre-training matters}: replacing SNN with GCN (w/o SNN) degrades performance, showing the importance of subgraph structural signals.
    \item \textbf{Ego-net extraction helps}: removing ego-net extraction (w/o ego net) hurts NDCG@10/20, especially for betweenness prediction, suggesting better node-centric structural focus.
\end{enumerate*}
\subsection{Pre-training Task Analysis}
\Cref{fig:pretrainingana} shows that \textbf{PreGress} achieves the best NDCG@10 on both local subgraph counting and centrality prediction, validating the effectiveness of our multi-task pre-training for downstream node ranking.
Among the variants, node-level ranking pretext tasks such as PageRank, eigenvector centrality, k-core, and degree centrality prediction are consistently stronger than link prediction, indicating a mismatch between edge-level meta-knowledge and node-level ranking targets.
\begin{wrapfigure}{r}{0.62\textwidth}
    \centering
    \includegraphics[width=\linewidth]{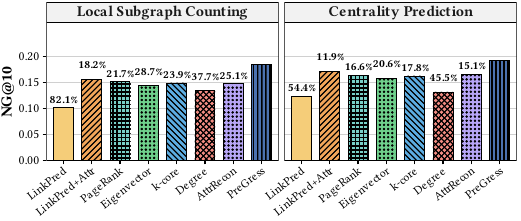}
    \caption{Performance of different pre-training tasks on WS dataset. Bars report NG@10, and labels above bars report performance difference.}
    \label{fig:pretrainingana}
\end{wrapfigure}
Moreover, degree centrality prediction combined with attribute reconstruction performs best, suggesting that structural ranking supervision and feature-level reconstruction provide complementary transferable knowledge for downstream node ranking.
Overall, these results suggest that node-centric pre-training tasks together with attribute reconstruction are more aligned with downstream ranking quality, answering Challenge \circledchar{1} in \Cref{sec:intro}.

\subsection{Rank-Distribution Alignment Analysis}
To further examine why degree centrality is a reasonable ranking-oriented pretext signal, we compare each structural signal as a zero-shot node ranker with downstream labels from betweenness centrality and local subgraph counting.
For local subgraph counting, we report the macro average over patterns and summarize alignment using rank correlation, proxy NDCG, top-\(k\) overlap, and JS distance.

\begin{figure}[htbp]
    \centering
    \includegraphics[width=\linewidth]{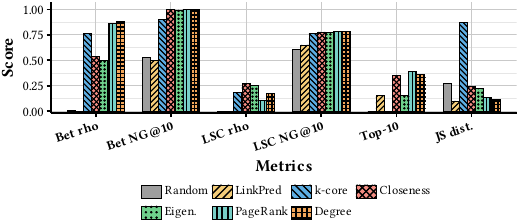}
    \caption{Zero-shot rank-distribution alignment on selected BK and YT datasets. Higher is better except for JS distance.}
    \label{fig:rank_alignment}
\end{figure}

\Cref{fig:rank_alignment} shows that node-level centrality signals are much more aligned with downstream ranking labels than random or edge-level link prediction.
Among the additional ranking-level pretext tasks, PageRank and eigenvector centrality provide strong betweenness alignment, while k-core is competitive for local subgraph counting.
Degree centrality remains among the best signals across both downstream targets and has the lowest computational cost, which explains why we use degree prediction together with attribute reconstruction as the final multi-task pre-training design.

\subsection{Over-smoothing Suppression Analysis}
\Cref{tab:oversmooth_ndcg} reports the betweenness centrality ranking performance in terms of NDCG@10 and NDCG@20 on WS as GNN depth increases. 
\begin{table}[t]
    \footnotesize
    \centering
    \caption{Betweenness centrality prediction performance (NDCG@10 / NDCG@20) with different GNN layers $l$ on WS dataset.}
    \label{tab:oversmooth_ndcg}
    \setlength{\tabcolsep}{5pt}
    \begin{tabular}{lcccccc}
        \toprule
        \multirow{2}{*}{\textbf{Method}} & \multicolumn{2}{c}{\textbf{$l$=10}} & \multicolumn{2}{c}{\textbf{$l$=20}} & \multicolumn{2}{c}{\textbf{$l$=30}} \\
        \cmidrule(lr){2-3} \cmidrule(lr){4-5} \cmidrule(lr){6-7}
         & \textbf{NG@10} & \textbf{NG@20} & \textbf{NG@10} & \textbf{NG@20} & \textbf{NG@10} & \textbf{NG@20} \\
        \midrule
        GCN         & 0.1298 & 0.1052 & 0.1044 & 0.0890 & 0.0846 & 0.0712 \\
        DeepGCN     & 0.1650 & 0.1441 & 0.1285 & 0.1007 & \underline{0.1087} & \underline{0.1045} \\
        \midrule
        SimGRACE    & \underline{0.1744} & 0.1543 & 0.1054 & 0.0859 & 0.0801 & 0.0719 \\
        GraphMAE    &0.1451 & 0.1272 & \underline{0.1484} & \underline{0.1113} & 0.0947 & 0.0755 \\
        GPPT & 0.1368& 0.1245 & 0.1163 & 0.1015 & 0.1029 & 0.0915 \\
        AdapterGNN  & 0.1588 & \underline{0.1545} & 0.1267 & 0.0089 & 0.0990 & 0.0851 \\
        \midrule
        GraphPrompt & 0.1403 & 0.1264 & 0.1181 & 0.0976 & 0.0807 & 0.0727 \\
        AllinOne & 0.1391 & 0.1385 & 0.1051 & 0.0972 & 0.0825 & 0.0651 \\
        PreGress    & \textbf{0.1763} &\textbf{0.1681} & \textbf{0.1665} & \textbf{0.1621} & \textbf{0.1812} & \textbf{0.1735} \\
        \bottomrule
    \end{tabular}
\end{table}
As the number of layers grows from 10 to 30, most baselines suffer from clear performance degradation, even with over-smoothing mitigation techniques such as residual connections (DeepGCN), self-supervised pre-training (SimGRACE, GraphMAE), or prompt-based adaptation (GraphPrompt, GPPT, AllinOne).
In contrast, PreGress maintains stable and even improved NDCG@10/20 at larger depths, achieving the best performance at $l{=}30$. This indicates that PreGress effectively suppresses over-smoothing and preserves discriminative node representations under deep architectures, leading to more reliable top-$k$ ranking quality.
\subsection{Few-shot Accuracy Evaluation}
\begin{table}[t]
    \footnotesize
    \centering
    \caption{Few-shot accuracy evaluation of betweenness centrality prediction task on WS dataset (NG@10 / NG@20).}
    \label{tab:fewshot_ndcg}
    \resizebox{\linewidth}{!}{
    \begin{tabular}{c|c|cc|cc|cc|cc}
    \hline
        \multirow{2}{*}{\textbf{Training Scheme}} & \multirow{2}{*}{\textbf{Method}} & 
        \multicolumn{2}{c|}{\textbf{10-shot}} & \multicolumn{2}{c|}{\textbf{20-shot}} & 
        \multicolumn{2}{c|}{\textbf{50-shot}} & \multicolumn{2}{c}{\textbf{full}} \\
        \cline{3-10}
        & & \textbf{NG@10} & \textbf{NG@20} & \textbf{NG@10} & \textbf{NG@20} & \textbf{NG@10} & \textbf{NG@20} & \textbf{NG@10} & \textbf{NG@20} \\
    \hline
    \hline
        \multirow{2}*{Finetuning} & SimGRACE & 0.1004 & 0.0967 & 0.1149 & 0.1058 & 0.1335 & 0.1241 & 0.1371 & 0.1341 \\
        & AdapterGNN & 0.1045 & \underline{0.1029} & 0.1168 & 0.1070 & 0.1285 & 0.1200 & 0.1313 & 0.1226 \\
    \hline
        \multirow{4}*{Prompt tuning} 
         & GraphPrompt & \underline{0.1145} & 0.1019 & \underline{0.1219} & 0.1043 & \underline{0.1505} & \underline{0.1476} & \underline{0.1662} & \underline{0.1872} \\
         & GPPT        & 0.1020 & 0.0981 & 0.1209 & \underline{0.1136} & 0.1308 & 0.1301 & 0.1344 & 0.1369 \\
         & AllinOne    & 0.1091 & 0.0714 & 0.1202 & 0.1134 & 0.1284 & 0.1277 & 0.1315 & 0.1314 \\
         & PreGress    & \textbf{0.1169} & \textbf{0.1154} & \textbf{0.1301} & \textbf{0.1247} &\textbf{ 0.1612} & \textbf{0.1547} & \textbf{0.1916} & \textbf{0.1879} \\
    \hline
    \end{tabular}
    }
\end{table}

\Cref{tab:fewshot_ndcg} reports the few-shot performance of betweenness centrality prediction on WS in terms of NDCG@10 and NDCG@20. Under all few-shot settings (10/20/50-shot), PreGress consistently outperforms fine-tuning and existing prompt-tuning baselines, with particularly clear margins in the extremely low-data regime. As the number of labeled nodes increases, PreGress exhibits steady improvements and approaches its full-data performance already at 50-shot. These results suggest that the pre-trained subgraph representations capture transferable ranking signals, while prompt-based adaptation effectively aligns limited supervision with the downstream node ranking objective.

\subsection{Efficiency Evaluation}
We evaluate query efficiency by comparing the average query time of PreGress with baseline methods across multiple datasets for both centrality prediction and local subgraph counting, and further analyze the efficiency--performance trade-off of different tuning strategies under the same pre-trained backbone.
\begin{figure}
    \centering
    \includegraphics[width=\linewidth]{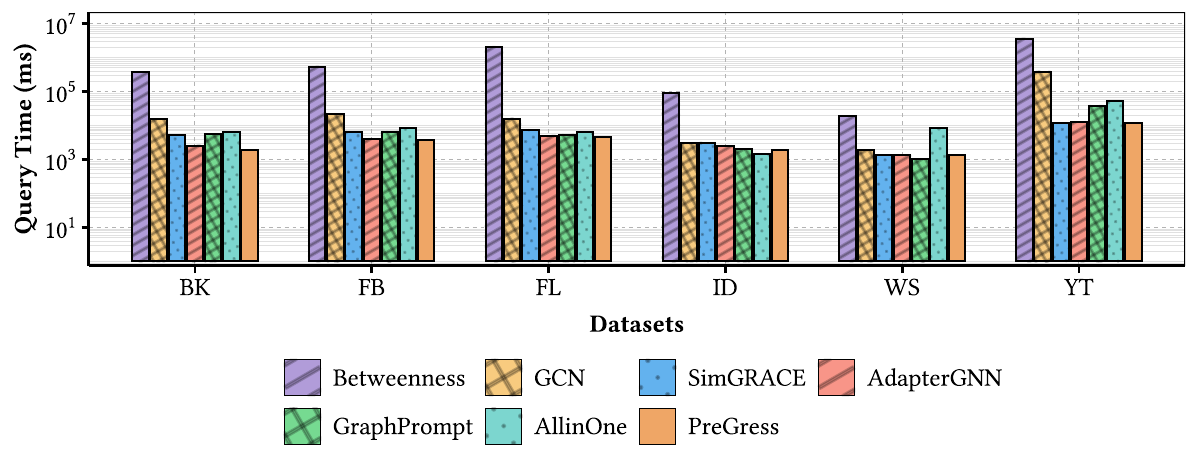}
    \caption{Average query time (ms) for betweenness centrality prediction.}
    \label{fig:effbet}
\end{figure}

\Cref{fig:effbet} reports the average query time for betweenness centrality prediction. Across all six datasets, PreGress consistently achieves low query latency, significantly outperforming the exact betweenness algorithm by several orders of magnitude. Compared with other learning-based methods, including fine-tuning and prompt-tuning approaches, PreGress remains among the fastest or achieves comparable efficiency, while delivering substantially better ranking quality. This demonstrates that PreGress provides an effective trade-off between accuracy and efficiency for large-scale node ranking.
\begin{figure}[h]
    \centering
    \includegraphics[width=\linewidth]{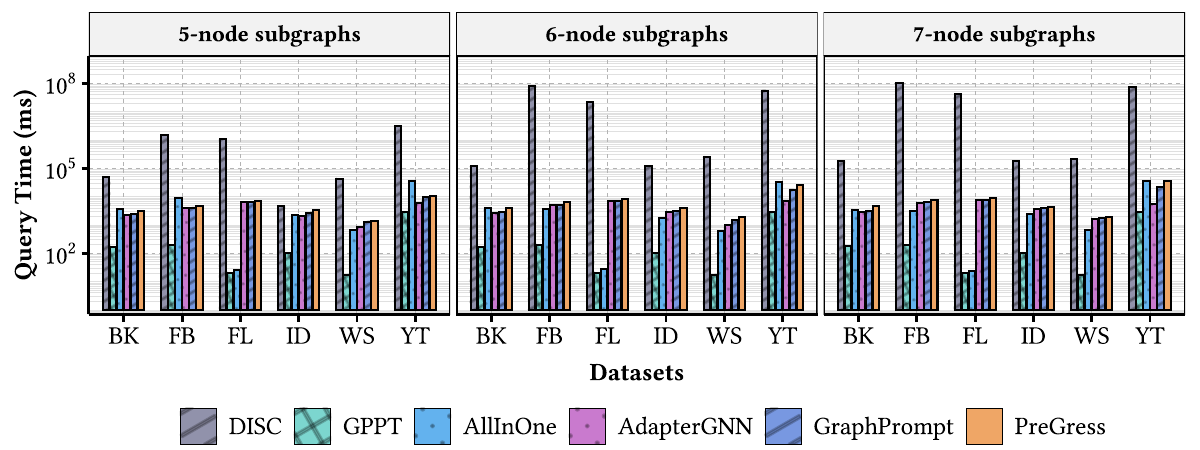}
    \caption{Average query time (ms) for local subgraph counting.}
    \label{fig:efflsc}
\end{figure}

\Cref{fig:efflsc} shows the efficiency results for local subgraph counting with different pattern sizes. PreGress consistently outperforms traditional counting methods such as DISC, achieving up to five orders of magnitude speedup. Among learning-based baselines, GPPT is faster on several datasets because it uses only lightweight token routing, but this comes with weaker ranking quality: its macro-average NDCG@10/20 is 0.271/0.267, compared with 0.347/0.322 for PreGress in \Cref{tab:lscacc}. Thus, PreGress is not intended to be the fastest neural baseline in every case; instead, it provides a better accuracy--efficiency trade-off by retaining orders-of-magnitude speedup over exact counting while delivering substantially stronger ranking performance.

\begin{table}[htbp]
\centering
\caption{Comparison of tuning strategies under the same pre-trained backbone on the local subgraph counting task.}
\label{tab:tuning_compare}
\resizebox{\linewidth}{!}{
\begin{tabular}{l|c|c|c|c|c}
\toprule
Method & Backbone Frozen & Trainable Params & Params Ratio & NDCG@10 & Peak Mem (GB) \\
\midrule
Full Fine-tuning & \ding{55} & 10241 & 100\% & 0.1508 & 3.48 \\
AdapterGNN & \ding{51} & 12289 & 36.23\% & 0.1313 & 0.59 \\
PreGress (Prompt) & \ding{51} & \textbf{1185} & \textbf{11.57\%} & \textbf{0.1916} & \textbf{0.20} \\
\bottomrule
\end{tabular}
}
\end{table}
Beyond query-time efficiency, \Cref{tab:tuning_compare} compares full fine-tuning, adapter-based tuning, and prompt-based tuning under the same backbone on local subgraph counting. Full fine-tuning updates all parameters but performs worse, suggesting negative transfer when directly adapting the backbone to downstream ranking. Adapter-based tuning improves efficiency by freezing the backbone, yet still incurs non-trivial trainable parameters and memory overhead.  In contrast, PreGress freezes the pretrained backbone and updates only the task-specific prompt generator and prediction head. This constrained adaptation preserves the pretrained ranking-oriented structural representation and reduces negative transfer, which is particularly important for local subgraph counting where labels are pattern-conditioned and highly skewed. As a result, prompt tuning achieves higher NDCG@10 than full fine-tuning, despite updating fewer parameters than the full model.
\subsection{Scalability Test} 
\begin{wrapfigure}{r}{0.56\textwidth}
    \centering
    \includegraphics[width=\linewidth]{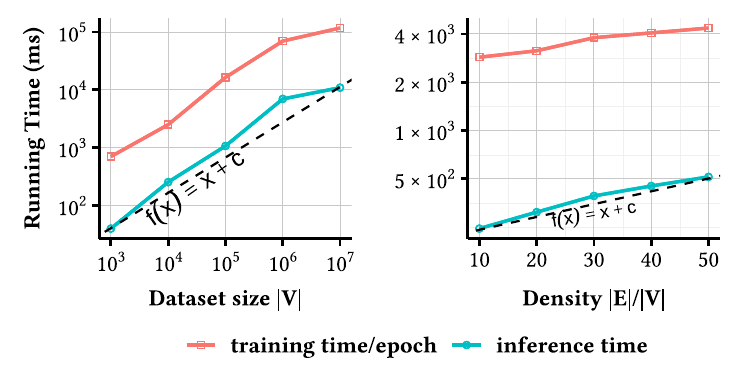}
    \caption{Scalability test on synthetic datasets.}
    \label{fig:scal}
\end{wrapfigure}
To evaluate the scalability of our approach, we conduct the experiment from two perspectives: dataset size and density. 
In the experiment relevant to data size, we use five synthetic datasets whose size range is \(10^3\) to \(10^7\) and whose density is set to 5. 
For the experiment relevant to density, five synthetic datasets with 10000 nodes and density from \(10\) to \(50\) are tested. The task applied in the scalability test is centrality centrality prediction. 
The findings presented in \Cref{fig:scal} illustrate that PreGress exhibits strong scalability and efficiency when handling large and dense graphs.

\section{Related Work}
\label{sec:related_work}
\textbf{Graph Node Ranking.} Graph node ranking is a fundamental problem in information retrieval \cite{guan2009personalized,narang2021ranking,narang2021ranking,geng2007feature,yu2015high,agarwal2006ranking}, 
which aims to induce an ordered list of nodes according to their relative importance or relevance.
Early studies \cite{guan2009personalized,yu2015high} such as GRoMO \cite{guan2009personalized} rely on handcrafted graph-theoretic measures to define ranking criteria, such as PageRank, betweenness centrality, and other structural importance scores.
More recent studies \cite{gupte2017role,gilpin2013guided} like RID$\epsilon$Rs \cite{gupte2017role} move beyond pairwise connectivity and incorporate higher-order structural evidence such as recurring local configurations and role-oriented patterns to define the centrality for ranking.
However, these approaches are often tightly coupled to a particular centrality and offer limited flexibility for handling diverse ranking criteria.
As graph retrieval scenarios increasingly involve different centrality measures \cite{du2023seq}, there is a growing need for scalable methods that learn a reusable scoring function for node ranking.

\textbf{Learning-based Approaches for Node Ranking.} 
To address the scalability limitations, learning-based methods~\cite{ergashev2023learning,he2022gnnrank,sankar2021graph} have been proposed to approximate node ranking scores or directly learn relevance scoring functions over graphs.
Despite their efficiency advantages, existing learning-based rankers require expensive ground-truth labels for each ranking criterion. 
Moreover, many methods~\cite{he2019fast,hu2008collaborative,qiu2018network} formulate node ranking as a regression problem and optimize pointwise losses, which does not necessarily align with ranking quality or preserve fine-grained ordering among nodes.
In addition, supervised GNN-based rankers \cite{he2022gnnrank} frequently exhibit limited generalization across graphs or ranking intents, and typically require retraining or extensive fine-tuning when the ranking criterion changes.
These challenges highlight the need for transferable ranking-oriented representations and efficient adaptation mechanisms for node ranking tasks.

\section{Conclusion}
\label{sec:conclusion}
In this paper, we introduce a novel prompt tuning-based pre-training approach designed to solve node ranking tasks.
By combining ranking-aligned pre-training objective with lightweight prompting, {\Method} offers a practical recipe for label-efficient and scalable deployment—pre-train once on easily accessible graph signals and quickly adapt to new ranking measures without full retraining.
Overall, {\Method} provides a simple yet transferable node scoring backbone for diverse ranking criteria under limited supervision.
Extensive experiments demonstrate the effectiveness and efficiency of our approach in node ranking tasks, highlighting its potential utility and applicability in further ranking scenarios.
We believe {\Method} takes a step toward unifying node ranking problems under a shared pre-training-and-prompting framework.

\bibliographystyle{ACM-Reference-Format}
\bibliography{main}


\end{document}